\documentclass[12pt,a4paper]{article}
\usepackage{amsmath, amssymb}
\usepackage{graphicx} %Loading the package
\usepackage{enumerate}
\usepackage{natbib}
\usepackage{url} % not crucial - just used below for the URL
\usepackage[affil-it]{authblk}  

\usepackage{indentfirst}

\usepackage{color, soul} %% pour surligner
\usepackage{bm} 
\usepackage{amsthm}
\usepackage{algorithm}

\usepackage{ulem}

\newcommand{\blind}{0}

\usepackage[switch, modulo]{lineno}

\newcommand{\R}{\mathbb{R}}
\newcommand{\N}{\mathbb{N}}

\newcommand{\E}{\mathbb{E}}
\newcommand{\p}{\mathbb{P}}
\newcommand\abs[1]{|#1|} %% absolute value
\newcommand{\norm}[1]{\lVert#1\rVert} %% norme
\DeclareMathOperator*{\rank}{rank}

\DeclareMathOperator*{\argmin}{arg\,min} %% argmin
\DeclareMathOperator*{\supp}{supp}
\DeclareMathOperator*{\sign}{sign}

\newcommand{\Xtilde}{\tilde{X}}
\newcommand{\betahat}{\hat{\beta}}
\newcommand{\gammahat}{\hat{\gamma}}
\newcommand{\omegahat}{\hat{\omega}}
\newcommand{\betatilde}{\tilde{\beta}}
\newcommand{\omegatilde}{\tilde{\omega}}
\newcommand{\betahatJP}{\hat{\beta}^{\operatorname{JP}}}
\newcommand{\betahatTJP}{\hat{\beta}^{\operatorname{TJP}}}
\newcommand{\omegahatJP}{\hat{\omega}^{\operatorname{JP}}}
\newcommand{\XNA}{X^{\textrm{NA}}}

\newcommand{\NA}{\texttt{NA}}

\newcommand{\lambdamax}{\lambda_{\max}}
\newcommand{\lambdamin}{\lambda_{\min}}
\newcommand{\sigmamin}{\sigma_{\min}}
\newcommand{\Sbar}{{\overline{S^0}}}
\newcommand{\Tbar}{{\overline{T^0}}}

\newtheoremstyle{thm}
{}
{}
{\itshape}
{}
{\bf}
{.}
{ }
{}

\newtheoremstyle{text}
{}
{}
{\upshape}
{}
{\bf}
{.}
{ }
{}

\newtheoremstyle{break}
{}
{}
{\upshape}
{}
{\bf}
{.}
{\newline}
{}

\newtheorem{theorem}{Theorem}

\newtheorem{lemma}{Lemma}
\newtheorem{prop}[theorem]{Proposition}
\newtheorem{corollary}[theorem]{Corollary}
\newtheorem*{theorem*}{Theorem}

\theoremstyle{text}
\newtheorem{remark}{Remark}
\newtheorem{definition}{Definition}

\theoremstyle{break}

\usepackage{color}
\usepackage[usenames,dvipsnames]{xcolor}
\usepackage[citecolor=Emerald,linkcolor=   NavyBlue,urlcolor=Magenta,colorlinks=true,pdfpagemode=UseOutlines]{hyperref}

\usepackage[colorinlistoftodos,textwidth=2.3cm]{todonotes}

\begin{document}

\def\spacingset#1{\renewcommand{\baselinestretch}%
	{#1}\small\normalsize} \spacingset{1.5}

%%%%%%%%%%%%%%%%%%%%%%%%%%%%%%%%%%%%%%%%%%%%%%%%%%%%%%%%%%%%%%%%%%%%%%%%%%%%%%

\if0\blind
{
	%\title{\bf Robust Lasso-Zero for the sparse corruption problem with application to model selection in the presence of missing covariates}
\title{\bf Robust Lasso-Zero for sparse corruption and model selection with missing covariates}
	%\author{Pascaline Descloux\thanks{pascaline.descloux@unige.ch}}
\author{Pascaline Descloux$^{(1)}$}
	\author{Claire Boyer$^{(2,3)}$}
	\author{Julie Josse$^{(4,5)}$}
	\author{Aude Sportisse$^{(6,7)}$}
	\author{Sylvain Sardy$^{(1)}$}
	
	\affil{\small $^{(1)}$ Department of Mathematics, University of Geneva, Switzerland \\
	$^{(2)}$ Sorbonne Université, CNRS, Laboratoire de Probabilités, Statistique et Modélisation (LPSM), F-75005 Paris, France, \\
	$^{(3)}$ MOKAPLAN, INRIA Paris \\
	$^{(4)}$ Inria centre
at Université Côte d’Azur, branch of Montpellier, France \\
	$^{(5)}$ Institut Desbrest d'Épidémiologie et de Santé Publique, Montpellier, France \\
	$^{(6)}$ Inria centre
at Université Côte d’Azur, Maasai, Valbonne, France \\
	$^{(7)}$ Interdisciplinary Institute for Artificial Intelligence, 3iA Côte d'Azur, France
	
	}
	
	\maketitle
} \fi

\if1\blind
{
	\bigskip
	\bigskip
	\bigskip
	\begin{center}
		{\LARGE\bf Title}
	\end{center}
	\medskip
} \fi

\bigskip

%\runningauthor{Pascaline Descloux, Claire Boyer, Julie Josse, Aude Sportisse and Sylvain Sardy}

\begin{abstract}
	We propose  Robust Lasso-Zero, an extension of the Lasso-Zero methodology, initially introduced for sparse linear models, to the sparse corruptions problem.    
We give  theoretical guarantees on the sign recovery of the parameters for a  slightly simplified version of the estimator, called Thresholded Justice Pursuit.
The use of Robust Lasso-Zero  is showcased for variable selection  with missing values in the covariates.   In addition to not requiring the specification of a model for the covariates, nor estimating their covariance matrix or the noise variance, the method has the great advantage of handling missing not-at random values without specifying a parametric model. 
Numerical experiments and a medical application underline the relevance of Robust Lasso-Zero in such a context with few available competitors. The method is easy to use and implemented in the R library {\tt lass0}. 
\end{abstract}

\vfill

\spacingset{1.5} % DON'T change the spacing!

\section{Introduction}

Consider the framework of sparse linear models in high dimension,

\begin{equation}\label{model:linear_initial}
y = X \beta^0 + \epsilon,
\end{equation}

where {$y\in \mathbb{R}^n$ is the vector of observations, $X\in\mathbb{R}^{n\times p}$ is the design matrix, and $\beta^0 \in \R^p$ is the parameter of interest, assumed to be $s$-sparse (only $s$ out of its $p$ entries are different from zero). The additive noise $\epsilon \in \R^n$ is a classical noise vector, qualified of ``dense noise" and assumed to be Gaussian with covariance matrix $\sigma^2I_n$}. 
In case of additional occasional corruptions, the sparse corruption problem can be then written as
\begin{equation} 
\label{model:sparsecorruption}
y = X \beta^0 + \sqrt{n}\omega^0 + \epsilon ,
\end{equation}
where  $\omega^0 \in \R^n$ is a $k$-sparse corruption vector;  see for instance \cite{chen2013}.
Noting that \eqref{model:sparsecorruption} can be rewritten as
\[
y =  \begin{bmatrix}X & \sqrt{n} I_n \end{bmatrix} \begin{bmatrix} \beta^0 \\ \omega^0 \end{bmatrix} + \epsilon,
\]
the sparse corruption model can be seen as a sparse linear model with an augmented design matrix and an augmented sparse vector of parameters.
We are interested in theoretical guarantees of support recovery for $\beta^0$ in \eqref{model:sparsecorruption}, 
with beneficial consequences for variable selection with missing covariates.

\paragraph{Related literature.} To recover $\beta^0$ {when there is no dense noise (when $\epsilon=0$ or equivalently when $\sigma^2 = 0$)}, several authors proposed the \textit{Justice Pursuit} (JP) program, name coined by \citet{laska2009}, by solving
\begin{equation} \label{opt:JP}
\begin{aligned}
&\min_{\beta \in \R^p, \omega \in \R^n} && \norm{\beta}_1 + \norm{\omega}_1
\\
& \quad \text{ s.t. } && y = X \beta +  \omega,
\\
\end{aligned}
\end{equation}
which is nothing else than the \textit{Basis Pursuit} (BP) problem, with the augmented matrix $\begin{bmatrix}X &  I_n \end{bmatrix}$ \citep{wright2009}.
\citet{wright2010} analyzed JP for Gaussian measurements, providing support recovery results when $n\simeq p$ using cross-polytope arguments. Besides, \citet{laska2009} and \citet{li2010systematic} proved that if the entries of~$X$ are  i.i.d.\ standard Gaussian as well, then the matrix $\begin{bmatrix} X &  I_n \end{bmatrix}$ satisfies some restricted isometry property with high probability, implying exact recovery of both $\beta^0$ and $\omega^0$, provided that $n\gtrsim (s+k)\log(p)$. However, in these works, the sparsity level $k$ of $\omega^0$ cannot be fixed to a proportion of the sample size $n$.
 Therefore, \citet{li2013} and \citet{nguyen2013b} introduced a tuning parameter $\lambda > 0$ and solve
\begin{equation} \label{opt:JPlambda}
\min_{\beta \in \R^p, \omega \in \R^n} \norm{\beta}_1 + \lambda \norm{\omega}_1
 \quad \text{ s.t. } \quad y = X \beta +  \omega.
\end{equation}
In a sub-orthogonal or Gaussian design, they both proved  exact recovery, even for  a large proportion of corruption. 

In the case of sparse and dense noise {(i.e. when $\omega^0\neq 0$ and $\sigma^2 > 0$)}, \citet{nguyen2013a} proposed to jointly estimate $\beta^0$ and $\omega^0$ by solving
\begin{equation} \label{opt:robust_lasso}
\min_{\beta \in \R^p, \omega \in \R^n} \frac{1}{2} \norm{y - X \beta - \omega}_2^2 + \lambda_\beta \norm{\beta}_1 + \lambda_\omega \norm{\omega}_1.
\end{equation}
In the special case where $\lambda_\beta = \lambda_\omega,$ problem~(\ref{opt:robust_lasso}) boils down to the Lasso \citep{tibshirani1996regression}  applied to the response $y$ and the design matrix $\begin{bmatrix} X & I_n \end{bmatrix}$. Assuming a standard Gaussian design and the invertibility and incoherence properties for the covariance matrix, they obtained sign recovery guarantee for an arbitrarily large fraction of corruption, provided that $n\geq C k \log (p) \log (n)$. In addition, the required number of samples is proven to be optimal. 
More recently in the case of a Gaussian design with an invertible covariance matrix,  \citet{dalalyan2019outlier} obtained an optimal rate of estimation of $\beta^0$ when considering an $\ell_1$-penalized Huber's $M$-estimator, which is actually equivalent to~(\ref{opt:robust_lasso}) \citep{Sardy01robustwavelet}.

\paragraph{Contributions.} To estimate the support of the parameter vector $\beta^0$ in the sparse corruption problem, we study an extension of the Lasso-Zero methodology \citep{descloux2018}, initially introduced for standard sparse linear models, to the sparse corruptions problem. 
We provide theoretical guarantees on the sign recovery of $\beta^0$ for a simplified version of the Robust Lasso-Zero, that we call the Thresholded Justice Pursuit (TJP). These guarantees are extensions of recent results on the Thresholded Basis Pursuit. 
The first one extends a result of \citet{tardivel2019sign}, providing a necessary and sufficient condition for consistent recovery in a setting where the design matrix is fixed but the nonzero absolute coefficients tend to infinity. The second one extends a result of \citet{descloux2018}, proving sign consistency for correlated Gaussian designs when $p,$ $s$ and $k$ grow with $n,$ allowing a positive fraction of corruptions. 

{In this paper, we also consider the case where the matrix of covariates $X$ contains missing values, which can be due to manual errors, poor calibration, insufficient resolution, etc. In the high-dimensional setting, note that the naive complete case analysis \citep{rubin1976inference}, which discards all incomplete rows, is not an option, because the missingness of a single entry causes the loss of an entire row, which contains a lot of information when $p$ is large \citep{zhu1996co}.}
Showing that missing values in the covariates can be reformulated into a sparse corruption problem, we recommend the Robust Lasso-Zero for dealing with missing data. 
For support recovery, this approach  requires neither to specify a model for the covariates or the missing data mechanism, nor an estimation of the covariates covariance matrix or of the noise variance, and hence provides a simple {(quasi hyperparameter-free)} method for the user . 
Numerical experiments and a medical application also underline the effectiveness of Robust Lasso-Zero with respect to few available competitors, {especially when the missing-data mechanism is MNAR (for Missing Not At Random), i.e.\ when data are said informative, since their own values can cause the missingness.}

\paragraph{Organization.} After defining the Robust Lasso-Zero methodology in Section ~\ref{sec:robustlasso}, we analyse the sign recovery properties of the Thresholded Justice Pursuit  in Section~\ref{sec:TJP}. Section~\ref{section:relation} is dedicated to variable selection with missing values and the  selection of tuning parameters is discussed in Section~\ref{section:qut}.
Numerical experiments are presented in Section~\ref{sec:numexp} and an application in  Section~\ref{sct:appli}.

\paragraph{Notation.}  %For a positive integer $p$, we d
Define $[p] := \{1, \ldots, p\},$ and the complement of a subset $S \subset [p]$ is denoted $\bar{S}.$ For a matrix $A$ of size $u \times v$ and a set $T \subset [v],$ we use $A_T$ to denote the submatrix of size $n \times \abs{T}$ with columns indexed by $T.$  We define the missing value indicator matrix $M \in \R^{n \times p}$ by $M_{ij} = \bm{1}_{\{\XNA_{ij} = \NA \}},$ and the set of incomplete rows by $\mathcal{M} := \{i \in [n] \ \mid \ M_{ij} = 1 \ \text{ for some } j \in [p] \}$.

\section{Robust Lasso-Zero} \label{sec:robustlasso}

\subsection{Lasso-Zero in a nutshell} \label{sec:lassozero}

Under the linear model~(\ref{model:linear_initial}), the Thresholded Basis Pursuit (TBP) estimates the parameter $\beta^0$ by  setting the small coefficients of the BP solution to zero. Since BP fits the observations $y$ exactly, noise is generally overfitted. The Lasso-Zero \citep{descloux2018} alleviates this issue by  solving repeated BP problems, respectively fed with the augmented matrices $[X  | G^{(k)}]$, where $G^{(k)} \in\R^{n \times n}$, $k = 1, \ldots, M,$ are different  i.i.d.\  Gaussian noise dictionaries {meaning that the columns of $G^{(k)}$ form a set of random vectors used to approximate the additive noise term in the observations. Note that each entry of $G^{(k)}$ is drawn according to a standard Gaussian distribution $\mathcal{N}(0,1)$}. {Once the noise dictionaries generated, } the corresponding obtained estimates $\betahat^{(1)}, \ldots, \betahat^{(M)}$ are then aggregated by taking the component-wise medians, further thresholded at level $\tau > 0.$
\citet{descloux2018} show that the Lasso-Zero tuned by Quantile Universal Thresholding \citep{giacobino2017} achieves a very good trade-off between high power and low false discovery rate compared to competitors.

\subsection{Definition of Robust Lasso-Zero} \label{sec:def_robust_lassozero}

{The Robust Lasso-Zero approach is an extension of the Lasso-Zero methodology for the sparse corruptions problem. }
Consider the sparse corruption model \eqref{model:sparsecorruption}, for which $S^0$ and $T^0$ denote the respective supports of $\beta^0$ and $\omega^0$, with $s := \abs{S^0}$ and $k := \abs{T^0}$ their respective sparsity degrees. 

To fix the notation, we then consider the following parameterization of Justice Pursuit (JP):
To fix the notation, we then consider the following parameterization of Justice Pursuit (JP):
\begin{equation} \label{opt:JP_newscaling}
 (\betahat^{\operatorname{JP}}_\lambda, \omegahat^{\operatorname{JP}}_\lambda) \in \argmin_{\beta \in \R^p, \ \omega \in \R^n} \norm{\beta}_1 + \lambda \norm{\omega}_1 \quad \text{s.t.} \quad y = X \beta + \sqrt{n} \omega.
\end{equation}
Renormalization by $\sqrt{n}$ balances the augmented design matrix $\begin{bmatrix}X & \sqrt{n} I_n \end{bmatrix}$:
in practice the columns of $X$ are often standardized so that $\norm{X_j}_2^2 = n$ for every $j \in [p],$ (or at least this is true in expectation) and in this way, the norms of all columns of $\begin{bmatrix}X & \sqrt{n} I_n \end{bmatrix}$ scale on the same order. 

{While the original Lasso-Zero algorithm was developed to solve standard sparse regression, and built upon solving a sequence of Basis Pursuit programs involving noise dictionaries, the Robust Lasso-Zero is designed to address the sparse corruptions setting, by solving a sequence of Justice Pursuit programs~\eqref{opt:JP_newscaling} involving noise dictionaries {(see \eqref{opt:extBP_robust}) and a median aggregation (see stage \ref{stage_Rlass0:median}) in Algorithm \ref{algo:robust-lass0})}. The Robust Lasso-Zero approach is fully described in Algorithm \ref{algo:robust-lass0}. } 
{In the latter, } attention has been paid to the estimation of the support of $\beta^0$. However the estimation of the corruption support is also possible by computing
the corresponding vectors $\omegahat^{\textrm{med}}_{\lambda}$ and $\omegahat^{\textrm{Rlass0}}_{(\lambda, \tau)}$, at stages~\ref{stage_Rlass0:median}) and~\ref{stage_Rlass0:threshold}).

\begin{algorithm}[h!]
Given data $(y, X)$, for given hyper-parameters $\lambda > 0, \tau \geq 0$ and $M \in \N^*:$
\begin{enumerate}[1)]
\item For $k = 1, \ldots, M:$
	\begin{enumerate}[i)]
	\item generate a matrix $G^{(k)}$ of size $n \times n$ with i.i.d. $\mathcal{N}(0,1)$ entries
	\item compute the solution $(\betahat^{(k)}_\lambda, \omegahat^{(k)}_\lambda, \gammahat^{(k)}_\lambda)$ to the augmented JP problem 
	\begin{equation} \label{opt:extBP_robust}
\begin{aligned}
&(\betahat^{(k)}_{\lambda}, \omegahat^{(k)}_{\lambda}, \gammahat^{(k)}_{\lambda}) \in &&\argmin_{\beta \in \R^p,\ \omega \in \R^n,\ \gamma \in \R^n} && \norm{\beta}_1 + \lambda \norm{\omega}_1 + \norm{\gamma}_1
\\
& && \quad \text{ s.t. } && y = X \beta + \sqrt{n} \omega + G^{(k)} \gamma.
\end{aligned}
\end{equation}
	\end{enumerate}
\item \label{stage_Rlass0:median} Define the vector $\betahat^{\textrm{med}}_\lambda$ by 
\[
\betahat^{\textrm{med}}_{\lambda, j} := \operatorname{median}\{\betahat^{(k)}_{\lambda, j}, k = 1, \ldots, M\} \quad \text{for every} \ j \in [p].
\]
\item \label{stage_Rlass0:threshold} Calculate the estimate
$\betahat^{\textrm{Rlass0}}_{(\lambda, \tau)} := \eta_\tau(\betahat^{\textrm{med}}_{\lambda})$,
where $\eta_{\tau}(x) = x \bm{1}_{(\tau, +\infty)}(\abs{x})$  hard-thresholds component-wise.
\end{enumerate}
%\end{algorithmic}
\caption{\label{algo:robust-lass0} Robust Lasso-Zero}
\end{algorithm}

Since the minimization problem \eqref{opt:extBP_robust} in Algorithm \ref{algo:robust-lass0} can be recast as a linear program, any relevant solver can be used (e.g., proximal methods).
Algorithm \ref{algo:robust-lass0} includes {three apparent hyper-parameters: the number $M$ of samplings, the regularization parameter $\lambda$ of \eqref{opt:JP_newscaling},  the thresholding parameter $\tau$ of the Robust Lasso-Zero methodology}. Their choice in practice is discussed in Section~\ref{section:qut}{, making the methodology quasi hyperparameter-free}.

\subsection{Theoretical guarantees on the Thresholded Justice Pursuit} \label{sec:TJP}

Discarding the noise dictionaries in Algorithm~\ref{algo:robust-lass0} amounts to thresholding the solution $(\betahat^{\operatorname{JP}}_\lambda, \omegahat^{\operatorname{JP}}_\lambda)$ of the Justice Pursuit problem~(\ref{opt:JP_newscaling}). The Robust Lasso-Zero can therefore be regarded as an extension of this simpler estimator, via \textit{Thresholded Justice Pursuit} (TJP):
\begin{equation}
\betahat^{\textrm{TJP}}_{(\lambda, \tau)} = \eta_\tau(\betahat^{\operatorname{JP}}_\lambda) \qquad {\rm and}\qquad 
\omegahat^{\textrm{TJP}}_{(\lambda, \tau)} = \eta_\tau(\omegahat^{\operatorname{JP}}_\lambda).
\end{equation}
We present two results about sign consistency of TJP. 
{Note that the statistical analysis of Robust Lasso-Zero methodology would be more mathematically involved, and is actually beyond the scope of this paper.}

\subsubsection{Identifiability as a necessary and sufficient condition for consistent sign recovery} \label{sec:ext_tardivel}

First introduced in \citet{tardivel2019sign} for the TBP, we propose the following extension of the identifiability notion for the TJP {characte\-ri\-zed by the uniqueness of the solution to JP.}
\begin{definition}
The pair $(\beta^0, \omega^0) \in \R^p \times \R^n$ is said to be \textit{identifiable} with respect to $X \in \R^{n \times p}$ and the parameter $\lambda > 0$ if it is the unique solution to JP~(\ref{opt:JP_newscaling}) when $y = X \beta^0 + \sqrt{n} \omega^0$.
\end{definition}

{In the case of JP, we provide the following characterization of identifiability, which depends on the sign vectors   $\sign(\beta^0)$ and $\sign(\omega^0)$.}

\begin{lemma} \label{lemma:identifiability_JP}
The pair $(\beta^0, \omega^0) \in \R^p \times \R^n$ is identifiable with respect to $X \in \R^{n \times p}$ and the parameter $\lambda > 0$ if and only if for every pair $(\beta, \omega) \neq (0, 0)$ such that $X \beta + \sqrt{n} \lambda^{-1} \omega = 0,$
\[
\abs{\sign(\beta^0)^T \beta + \sign(\omega^0)^T \omega} < \norm{\beta_{\Sbar}}_1 + \norm{\omega_{\Tbar}}_1.
\]
\end{lemma}
\begin{proof}
See Appendix \ref{app:proof_thm_ext_tardivel}. 
\end{proof}

{Note that recently, \citet{schneider2020geometry} provide a necessary and sufficient condition for the uniqueness of the solution to BP and show that the set of design matrices which do not satisfy this condition is negligible with respect to the Lebesgue measure on $\R^{n\times p}$.}

\begin{remark}[Identifiability of JP vs. BP]
{Let us observe that $(\beta_0,\omega_0)$ is identifiable for JP w.r.t.\ $X$ and $\lambda$ then $\beta_0$ is identifiable for BP w.r.t.\ $X$ (i.e.\ in the noiseless case and without corruptions, $\beta_0$ can be recovered by solving BP). Moreover, when the observations are not corrupted (i.e.\ when $y = X \beta_0 + \epsilon$), the identifiability of $\beta_0$ for BP w.r.t.\ $X$ is actually a necessary and sufficient condition for sign recovery by thresholded BP (see \cite{tardivel2019sign}). Consequently, this corroborates the  intuitive idea that sign recovery is a more difficult task via thresholded JP in the corrupted case than via thresholded BP in the non-corrupted case.}
\end{remark}

In order to show that identifiability is necessary and sufficient for TJP to consistently recover $\sign(\beta^0)$ and $\sign(\omega^0)$, {we consider} a fixed matrix $X \in \R^{n \times p}$ and a sequence {of parameters} $\{ (\beta^{(r)}, \omega^{(r)})\}_{r \in \N^*}$ {such that} the following holds:
\begin{enumerate}[(i)]
\item \label{ass2:sign_invariance} there exist sign vectors $\theta \in \{1, -1, 0\}^p$ and $\tilde{\theta} \in \{1, -1, 0\}^n$ such that $\sign(\beta^{(r)}) = \theta$ and $\sign(\omega^{(r)}) = \tilde{\theta}$ for every $r \in \N^*,$
\item \label{ass2:inf_betamin} $\lim_{r \to +\infty} \min\{\beta^{(r)}_{\min}, \omega^{(r)}_{\min}\} = +\infty,$ {where $\beta_{\min}^{(r)}:= \min_{j \in \supp(\beta)} \abs{\beta_j^{(r)}},$ and $\omega_{\min}^{(r)}:= \min_{j \in \supp(\omega)} \abs{\omega_j^{(r)}}$.}
\item \label{ass2:bounded_ratio} there exists $q > 0$ such that $\frac{\min\{\beta^{(r)}_{\min}, \omega^{(r)}_{\min}\}}{\max\{\norm{\beta^{(r)}}_{\infty}, \norm{\omega^{(r)}}_{\infty}\}} \geq q.$ 
\end{enumerate}
{Note that} these assumptions are similar to the ones of \citet{tardivel2019sign}. {Assumption \eqref{ass2:sign_invariance} requires that each parameter of the sequence $\{ (\beta^{(r)}, \omega^{(r)})\}_{r \in \N^*}$ have invariant sign vectors. 
For a fixed design matrix and a fixed $\lambda$, Assumptions \eqref{ass2:inf_betamin} and \eqref{ass2:bounded_ratio} are asymptotic statements, meaning that the nonzero absolute components of the sequence parameters $\{ (\beta^{(r)}, \omega^{(r)})\}_{r \in \N^*}$ tend to infinity in a certain sense.}
We use the notation $S^0 := \supp(\theta) = \supp(\beta^{(r)})$ and $T^0 := \supp(\tilde{\theta}) = \supp(\omega^{(r)}).$ {For each parameter $(\beta^{(r)}, \omega^{(r)}), r\in \mathbb{N}^*$, let us consider the sparse corruption problem  $y^{(r)} := X \beta^{(r)} + \sqrt{n} \omega^{(r)} + \epsilon$.} 
We denote by $(\betahat^{\operatorname{JP}(r)}_{\lambda}, \omegahat^{\operatorname{JP}(r)}_{\lambda})$ the JP solution when $y = y^{(r)}$ and $(\betahat^{\operatorname{TJP}(r)}_{(\lambda, \tau)}, \omegahat^{\operatorname{TJP}(r)}_{(\lambda, \tau)})$ the corresponding TJP estimates.

\begin{theorem} \label{thm:ext_tardivel}
{
Let $\lambda > 0$ and let $X$ be a matrix of size $n \times p$ such that for any $y \in \R^n,$ the solution to JP~(\ref{opt:JP_newscaling}) is unique.

\textbf{Necessary condition}: If there exists $\tau \geq 0$ such that 
$$
\sign (\hat{\beta}^{TJP}_{\lambda, \tau}) = \sign (\beta_0) \qquad \text{and} \qquad \sign (\hat{\omega }^{TJP}_{\lambda, \tau} ) =\sign (\omega_0)
$$
then $(\beta_0,\omega_0)$ is identifiable with respect to $X$ and $\lambda$.

\textbf{Sufficient condition} (asymptotically): Let $\{(\beta^{(r)}, \omega^{(r)})\}_{r \in \N^*}$ be a sequence satisfying Assumptions~\eqref{ass2:sign_invariance}-\eqref{ass2:bounded_ratio} above. If the pair of sign vectors $(\theta, \theta')$ is identifiable with respect to $X$ and $\lambda$, then for
every $\epsilon\in\mathbb{R}^n$, there exists $R = R(\epsilon) > 0$ such that for every $r \geq R$ there is a threshold $\tau = \tau(r) > 0$ for which
\begin{equation} \label{TJP_equal_signs}
\sign (\hat{\beta}^{TJP}_{\lambda, \tau}) = \theta \qquad \text{and}
\qquad \sign(\hat{\omega }^{TJP}_{\lambda, \tau}) = \theta'.
\end{equation}
}
\end{theorem}

\begin{proof}
See Appendix \ref{app:proof_thm_ext_tardivel}. 
\end{proof}

\begin{remark}
One might be interested in recovering the signs of the sparse corruption. If $\omega^{(r)}$ is considered as noise, then only the recovery of $\sign(\beta^{(r)})$ matters. In this case one could weaken Assumptions~\eqref{ass2:inf_betamin} and~\eqref{ass2:bounded_ratio} above by replacing $\min\{\beta^{(r)}_{\min}, \omega^{(r)}_{\min}\}$ by $\beta^{(r)}_{\min},$ and identifiability of $(\theta, \tilde{\theta})$ would be sufficient for recovering $\sign(\beta^0).$ However, recovery of both $\sign(\beta^{(r)})$ and $\sign(\omega^{(r)})$ is needed for proving necessity of identifiability.
\end{remark}

Identifiability of sign vectors is necessary and sufficient for sign recovery when the nonzero coefficients are large. 
In the next section, we particularize these results to Gaussian designs,
and prove that sign consistency holds,
allowing $p, s$ and $k$ to grow with the sample size $n$, in a classical way.

\subsubsection{Sign consistency of TJP for correlated Gaussian designs} \label{sec:sign_consistency_gaussian}
We make the following assumptions:
\begin{enumerate}[(i)]
\setcounter{enumi}{3}
\item \label{assumption:Gaussdesign} the rows of $X \in \R^{n \times p}$ (with $n < p$) are random and i.i.d. $\mathcal{N}(0, \Sigma)$;
\item \label{assumption:boundedsigmamin} The smallest eigenvalue of the covariance matrix $\Sigma$ is assumed to be positive: $\lambda_{\min}(\Sigma) > 0$,
\item \label{assumption:var1} the variance of the covariates is equal to one: $\Sigma_{ii}=1$ for every $i \in [p]$;
\item \label{assumption:Gaussiannoise} the noise is assumed to be Gaussian $\epsilon \sim \mathcal{N}(0, \sigma^2 I_n)$ {and independent of the columns of $X$.}
\end{enumerate}
Assumptions~\eqref{assumption:Gaussdesign} and \eqref{assumption:boundedsigmamin} imply that  almost surely $\rank{X} = n$.
\begin{prop}\label{thm:TJP} 
Under Assumptions~\eqref{assumption:Gaussdesign}-\eqref{assumption:Gaussiannoise}, choosing
$
\lambda = \frac{1}{\sqrt{\log{p}}}
$
ensures with probability greater than $1 - ce^{-c'n} - 1.14^{-n} - 2e^{-\frac{1}{8}(\sqrt{p} - \sqrt{n})^2}$, that there exists a value of $\tau > 0$ such that 
\[
\sign (\betahat^{\operatorname{TJP}}_{(\lambda, \tau)} ) = \sign (\beta^0),
\]
provided that 
\begin{align}
&n \geq C \frac{\kappa (\Sigma) }{\lambda_{\min}(\Sigma)} s \log{p},& \label{lb:n_TJP}
\\
&\frac{n}{k} \geq \max\left\{ \frac{1}{C'},\frac{\kappa (\Sigma)}{C'' } \right\}  \label{ub:k0},
\\
&\beta^0_{\min} >  \frac{10 \sqrt{2}\max\{1, \lambda\}  \sigma  \sqrt{p+n}}{\left(\frac{\lambdamin(\Sigma)}{4}(\sqrt{p/n} - 1)^2 + 1 \right)^{1/2}},& \label{beta_min}
\end{align}
where $\kappa(\Sigma):=\frac{\lambdamax(\Sigma)}{\lambdamin(\Sigma)}$ is the conditioning number of $\Sigma$, and $C, C', C''$ are some numerical constants with $C \geq 144^2.$
\end{prop}
\begin{proof}
See Appendix \ref{app:proof_thm_TJP}.
\end{proof}

Note that the lower-bound required on $\beta^0_{\min}$ in Proposition \ref{thm:TJP} is {strong but complies with the one required for TBP in}  \citep{descloux2018}. {We believe that this limitation, which is not observed in simulations, is a proof artefact that could be relaxed.} {Besides, the advantage of noise dictionaries used in the Lasso-Zero in \citep{descloux2018} (and thus for the Robust Lasso-Zero) is twofold: first it allows to better separate signal from noise, second it allows a form of resampling which, in the spirit of stability selection, allows to better identify needles in the haystack.}

Note also that the conditioning number of $\Sigma$ comes into play in the lower-bounds required on $n$ and $k$. This quantity indeed seems natural to arise in the sparse corruption problem helping discriminating design instability from corruptions.

{Corollary \ref{theo:add_gauss} trivially follows from Proposition \ref{thm:TJP} and ensures that, for correlated Gaussian designs and signal-to-noise ratios high enough, if $\Sigma$ is well-conditioned, TJP succesfully recovers $\sign (\beta^0)$ wih high probability, even with a positive fraction of corruptions. }
\begin{corollary} \label{theo:add_gauss} {Let $X \in \mathbb{R}^{n\times p}$ be a Gaussian matrix satisfying Assumptions~\eqref{assumption:Gaussdesign}-\eqref{assumption:Gaussiannoise}. Assume also that the eigenvalues of covariance matrix $\Sigma$ are bounded $0 < \gamma_1 \leq \lambdamin(\Sigma) \leq \lambdamax(\Sigma) \leq \gamma_2$.

\textbf{Sufficient condition}: If $p/n \to \delta > 1$, one has
$$\mathbb{P}\left(\exists \tau>0, \sign (\betahat^{\operatorname{TJP}}_{(\lambda, \tau)} ) = \sign (\beta^0)\right) \geq 1 - (c+2e^{-\frac{1}{8}p})e^{-c'n},$$
with $c$ and $c'$ some numerical constants and provided that
$n = \Omega(s\log{p}),$ $k = \mathcal{O}(n)$ and $\beta^0_{\min} = \Omega(\sqrt{n})$}\footnote{{Recall that $f(n)=\Omega(g(n))$ (resp. $f(n)=\mathcal{O}(g(n))$) means that there exists $\kappa>0$ (resp. $\kappa'>0$) such that $|f(n)|\geq \kappa|g(n)|$ (resp. {$|f(n)|\leq \kappa'|g(n)|$})  for large $n$ enough.}}.
\end{corollary}

\color{black}

\section{Model selection with missing covariates} 
 \label{sec:application_to_NA}

In practice the matrix of covariates $X$ is often partially known and one only observes an incomplete matrix, denoted $\XNA$. \footnote{\texttt{NA}, for Not  Available, is a usual symbol for missing values.}
High dimensional variable selection with missing values turns out to be a challenging problem and very few solutions are available, not to mention implementations. Available solutions either require strong assumptions on the missing value mechanism, a lot of parameters  tuning or strong assumption on the covariates distribution which is hard in high dimensions. They include the Expectation-Maximization algorithm  \citep{dempster1977maximum} for sparse linear regression \citep{garcia2010variable}
and regression imputation methods \citep{van2018flexible}.

A method combining penalized regression techniques with multiple imputation and stability selection has been developed \citep{liu2016variable}. Yet, aggregating different models for the resulting multiple imputed data sets becomes increasingly complex as the number of data grows. 
 \citet{rosenbaum2013improved} modified the Dantzig selector by using a consistent estimation of the design covariance matrix. Following the same idea, \citet{loh2012} and \citet{datta2017} reformulated the Lasso also using an estimate of the design covariance matrix, possibly resulting in a non-convex problem.
\citet{chen2013noisy} presented a variant of orthogonal matching pursuit which recovers the support and achieves the minimax optimal rate. 
 \citet{jiang2019adaptive} proposed Adaptive Bayesian SLOPE, combining SLOPE and Spike-and-Slab Lasso. 
{While some of these methods have interesting theoretical guarantees,} they all require an estimation of the design covariance matrix, which is often obtained under the restrictive MCAR {(Missing Completely At Random)} assumption{, when the missingness does not depend on the data}.

\subsection{Relation to the sparse corruption model}  \label{section:relation}
 
To tackle the problem of  estimating the sparse model parameter  when the design matrix is incomplete, we suggest an easy-to-implement solution for the user, which consists in imputing the missing entries in $\XNA$ with the imputation of his choice to get a completed matrix $\tilde X$, {such as the mean imputation}, and to take into account the impact of the occasional poor imputation as follows.  Given the matrix $\tilde X$, the linear model~(\ref{model:linear_initial}) can be rewritten in the form of the sparse corruption model \eqref{model:sparsecorruption}, where $\omega^0 := \frac{1}{\sqrt{n}}(X - \tilde X) \beta^0$ is the (unknown) corruption due to imputations. In classical (i.e. non-sparse) regression, one could not say much about $\omega^0$ without any prior knowledge of the distribution of the covariates or the missing data mechanism. Since the key point here is that when $\beta^0$ is sparse, then so is $\omega^0,$ even if all rows of the design matrix contain missing entries. Indeed, for every $i \in [n],$
\begin{equation} \label{eq:omega0}
\omega^0_i = \frac{1}{\sqrt{n}} \sum_{j = 1}^p (X_{ij} - \tilde X_{ij}) \beta^0_j  = \frac{1}{\sqrt{n}} \sum_{j \in S^0} (X_{ij} - \tilde X_{ij}) \beta^0_j,
\end{equation}
so $\omega^0_i$ is nonzero only if the $i^{\textrm{th}}$ row of $\XNA$ contains missing value(s) on the support $S^0$, since $(X_{ij} - \tilde X_{ij}) = 0$ if $X_{ij}$ is observed.
So  the problem of missing covariates can be rephrased as a sparse corruption problem,  as already pointed out in \cite{chen2013robust}.
We propose to use  the Robust Lasso-Zero  methodology presented in Section~\ref{sec:def_robust_lassozero}, {specifically designed to solve the sparse corruption problem}, see Algorithm \ref{algo:NAlass0}.

Note that if the $i^{\text{th}}$ row of $X$ is fully observed, then $\omega^0_i = 0$ by~(\ref{eq:omega0}). Thus the dimension of $\omega^0$ can be reduced by restricting it to the incomplete rows of $\XNA.$ The corruption vector $\omega^0$ is now of size $\abs{\mathcal{M}}$ and~(\ref{model:sparsecorruption}) becomes 
\begin{equation} \label{model:NAsparsecorr}
y = X \beta^0 + \sqrt{n} I_{\mathcal{M}}\omega^0 + \epsilon,
\end{equation}
{where $I_{\mathcal{M}} \in \mathbb{R}^{n\times |\mathcal{M}|}$ is the submatrix of the idendity matrix $I_n \in \mathbb{R}^{n\times n}$ with columns indexed by $i \in \mathcal{M}$.}

\begin{algorithm}
Given data $(y,\XNA)$, for given hyper-parameters $\lambda > 0, \tau \geq 0$ and $M \in \N^*$:
\begin{enumerate}[1)]
\item Impute $\XNA$ 
and rescale the imputed matrix $X$ such that all columns have Euclidean norm equal to $\sqrt{n}$.
\item Run Algorithm \ref{algo:robust-lass0} with the design matrix~$X$.
\end{enumerate}
\caption{Robust Lasso-Zero for missing data \label{algo:NAlass0}}
\end{algorithm}

\subsection{Selection of tuning parameters} \label{section:qut}

Formally, Algorithm~\ref{algo:NAlass0} requires {three apparent} hyper-parameters: $M$, $\lambda$ and $\tau$.

{The parameter $M$ is the number of noise dictionaries also used to perform the median estimation. It controls the precision in estimating the marginal medians and does not play an important role: it can be seen as an analogue of the number of bootstrap samples in a bootstrap procedure. As a matter of fact, as it exists rules of thumb for bootstrap samples, we recommend $M=20$ -as we used in our experiments- providing satisfactory precision. However, it is important to take $M$ larger than one (see Figure 3 in \citep{descloux2018} for instance).

The parameter $\lambda$ in \eqref{opt:JP_newscaling} tunes the balance between the corruption regularization and the model parameter one. For a fair isotropic penalty on $\beta,\omega$ and $\gamma$, we fix $\lambda=1$ (even if Proposition \ref{thm:TJP} and Corollary \ref{theo:add_gauss} do not provide theoretical guarantees for $\lambda = 1$).

Finally, the thresholding parameter $\tau$ can be selected through the Quantile Universal Threshold (QUT) \citep{giacobino2017} methodology,  generally driven by model selection rather than prediction. Indeed, 
}
under the null model,  no sparse corruption exists: indeed if $\beta^0 = 0$, so is $\omega^0$ since  $\omega^0= \frac{1}{\sqrt{n}} (X - \tilde X) \beta^0 = 0$. QUT selects the tuning parameter so that under the null model ($\beta^0 = 0$), the null vector $\betahat = 0$ is recovered with probability $1-\alpha$. 
Under the null model,   $y = \epsilon$  whatever the missing data pattern is. 
Then given a fixed value of $\lambda$ and a fixed imputed matrix $\tilde X$, the corresponding QUT value of $\tau$ is the upper $\alpha$-quantile of $\norm{\betahat^{\textrm{med}}_{\lambda}(\epsilon)}_\infty,$ where $\betahat^{\textrm{med}}_{\lambda}(\epsilon)$ is the vectors of medians obtained at stage~\ref{stage_Rlass0:median}) of Algorithm~\ref{algo:robust-lass0} applied to $\tilde X$ and $y = \epsilon.$ To free ourselves from preliminary estimation of the noise level $\sigma,$ we exploit the noise coefficients $\gammahat^{(k)}$ of Robust Lasso-Zero to pivotize the statistic $\norm{\betahat^{\textrm{med}}_{\lambda}(\epsilon)}_\infty,$ as explained in \citet{descloux2018}. Thus, the parameter $\tau$, through the QUT methodology, is depending on $\lambda$ (that we choose equal to 1) and the FDR $\alpha$ under the null hypothesis. As usual, one can choose a priori a FDR of $\alpha=0.05$, or smaller depending on the application. 

Despite all these apparent hyperparameters, the method can be still considered as relying on mild tuning as the only important choice is the choice of $\tau$ by the QUT methodology.

\section{Numerical experiments} \label{sec:numexp}

We evaluate the performance of the Robust-Lasso Zero when missing data affect the design matrix. The code reproducing these experiments is available at \url{https://github.com/pascalinedescloux/robust-lasso-zero-NA}.

\subsection{Simulation settings}

\paragraph{Simulation scenarios.} 

We generate data according to model ~\eqref{model:linear_initial}
with the covariates matrix obtained by drawing {$n=100$} observations from a Gaussian distribution $\mathcal{N}(0, \Sigma)$, where $\Sigma\in \mathbb{R}^{200\times 200}$ is a Toeplitz matrix, such that $\Sigma_{ij} = \rho^{\abs{i-j}}$;  the variance of the noise $\sigma = 0.5$ and  the coefficient $\beta^0$ are drawn uniformly  from $\{\pm 1\}$.
We make the following parameters vary: 
\begin{itemize}
\item Correlation structures indexed by $\rho$ with $\rho=0$ (uncorrelated) and $\rho=0.75$ (correlated);
\item Sparsity degrees indexed by $s$ with $s \in \{3,10\}$. 
\end{itemize}
Before generating the response vector $y$, all columns of $X$ are mean-centered and standardized.
Missing data are then introduced in $X$ according to two different mechanisms, MCAR or MNAR, and  
in two different proportions. {Recall that the data are said (i) MCAR if the missingness does not depend on the data values and (ii) MNAR if the missingness depends on the data values.} More precisely, any entry of $X$ is missing according to the following logistic model
\[
\p(\XNA_{ij} = \NA \ \mid \ X_{ij} = x) = \frac{1}{1 + e^{-a \abs{x} - b}},
\]
where $a \geq 0$ and $b \in \R$. Choosing $a = 0$ yields MCAR data, whereas  $a = 5$ leads to MNAR setting in which high absolute entries are more likely to be missing. For a fixed $a$,  the value of $b$ is chosen so that the overall average proportion of missing values is $\pi,$ with $\pi = 5\%$ and $\pi = 20\%.$

Two sets of simulations are run. The first one is  ``$s$-oracle",  meaning that the tuning parameters of the different methods are chosen so that the estimated support has correct support size~$s$. In the second set,  no knowledge of $s, \beta^0$ or $\sigma$ is provided.

\paragraph{Performance evaluation.} The performance of each coming estimator is assessed in terms of the following criteria, averaged over 100 replications:
\begin{itemize}
\item  the Probability of Sign Recovery (PSR), 
$ \mathrm{PSR} = \p(\sign (\betahat) {=} \sign (\beta^0)),$
\item  the signed True Positive Rate (sTPR),  where   $\operatorname{s-TPR} = \E(\operatorname{s-TPP}) \quad {\rm with}$ 
\begin{equation} \label{def:sTPP}
\operatorname{s-TPP} := \frac{\abs{\{j \ \mid \ \beta^0_j > 0, \betahat_j > 0\}} + \abs{\{j \ \mid \ \beta^0_j < 0, \betahat_j < 0\}}}{\abs{S^0}},
\end{equation}
which is the proportion of nonzero coefficients whose sign is correctly identified;
\item the signed False Discovery Rate (sFDR): $\operatorname{s-FDR} = \E(\operatorname{s-FDP}) \quad {\rm with}$ 
\begin{equation} 
 \label{def:sFDP}
\operatorname{s-FDP} := \frac{\abs{\hat{S}} - \abs{\{j \ \mid \ \beta^0_j > 0, \betahat_j > 0\}} - \abs{\{j \ \mid \ \beta^0_j < 0, \betahat_j < 0\}}}{\max\{1, \abs{\hat{S}}\}},
\end{equation}
which is the proportion of incorrect signs among all discoveries.
\end{itemize}

\paragraph{Estimators considered.}
We compare the following estimators: 
\begin{itemize}
\item \textbf{Rlass0:} the Robust Lasso-Zero described in Algorithm~\ref{algo:NAlass0} using {the mean imputation} and $M$ equal to 30. The tuning parameters are obtained using arbitrarily $\lambda = 1$ and selecting $\tau$  by quantile universal threshold (QUT) at level $\alpha = 0.05$.
\item \textbf{lass0:}  the Lasso-Zero proposed in \citet{descloux2018}. The automatic tuning is performed by QUT, at level $\alpha = 0.05$.
\item \textbf{lasso:} the Lasso \citep{tibshirani1996regression} performed on the mean-imputed matrix where the regularization parameter is tuned by cross-validation.
\item \textbf{NClasso:} the nonconvex $\ell_1$ estimator of \citet{loh2012}. It is  only  included under the $s$-oracle setting, as selection of the tuning parameter in practice is not discussed in their work.
\item \textbf{ABSLOPE:} Adaptive Bayesian SLOPE of \citet{jiang2019adaptive}. 
\end{itemize}
{
To ease the readability of the numerical results, we also provide in Appendix \ref{sec:othernumexp_othestim} additional comparisons to other estimators: the thresholded Robust Lasso proposed in \cite{nguyen2013a} and the thresholded lasso in \cite{pokarowski2019improving}. Rlass0 still remains competitive in terms of sign recovery, in particular in difficult cases, i.e.\ when the percentage of missing values increases, when the missing data are informative and when the covariates are correlated. For the cases where Rlass0 has a lower probability of sign recovery than the other methods, it still remains competitive in terms of s-FDR and outperforms the thresholded methods in terms of s-TPR. 

In the following simulations, we impute each missing variable by its empirical mean computed with the observed individuals (as practitioners would tend to do). 
We refer to Appendix \ref{sec:othernumexp_imp} in which we provide extra numerical experiments when different strategies of imputation are considered, without changing the conclusions observed in this section. 

}

\subsection{Results}

\subsubsection{With $s$-oracle hyperparameter tuning} \label{section:soracle_tuning}

Under the $s$-oracle tuning, an $\operatorname{s-TPP}$~(\ref{def:sTPP}) of one means that the signs of~$\beta^0$ are exactly recovered, and the $\operatorname{s-TPP}$ is related to the $\operatorname{s-FDP}$~(\ref{def:sFDP}) through $\operatorname{s-FDP} = 1 - \operatorname{s-TPP}$. That is why, in Figure \ref{fig:soracle}, only  the average $\operatorname{s-TPP}$ and the estimated probability of sign recovery are reported.  

\paragraph{Small missingness -- High sparsity ($5\%$ of NA and $s=3$).}
In the non-correlated case, in Figure \ref{fig:soracle} (a) and (c),  MCAR and MNAR results are similar across methods. 
With correlation, in Figure\ref{fig:soracle} (b) and (d), Rlass0 improves PSR and sTPR, specially with MNAR data.

\begin{figure}[h!]
\vspace{-2cm}
\hspace{-2cm}
\begin{tabular}{cc}
\includegraphics[scale=0.5]{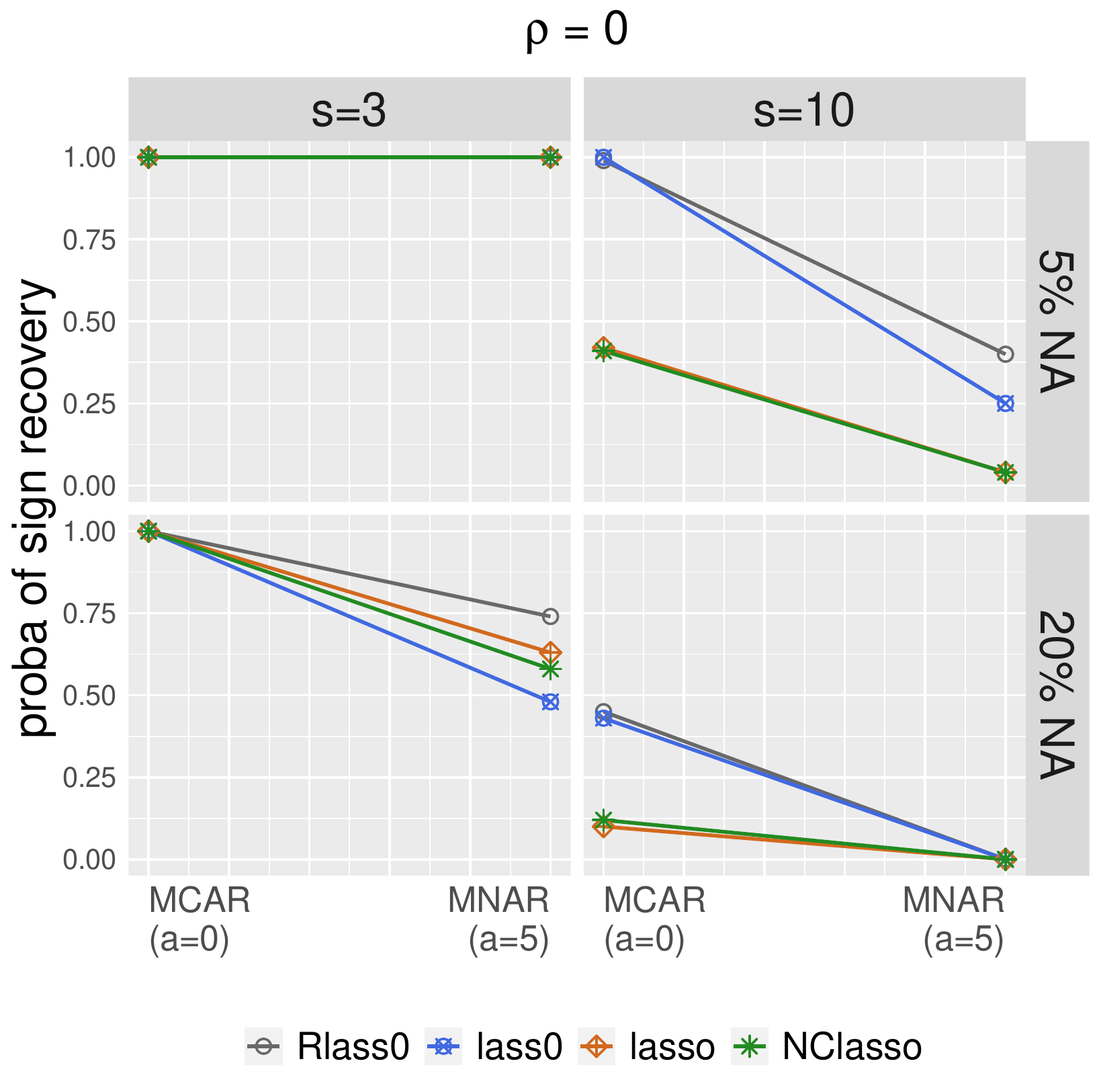} &
\includegraphics[scale=0.5]{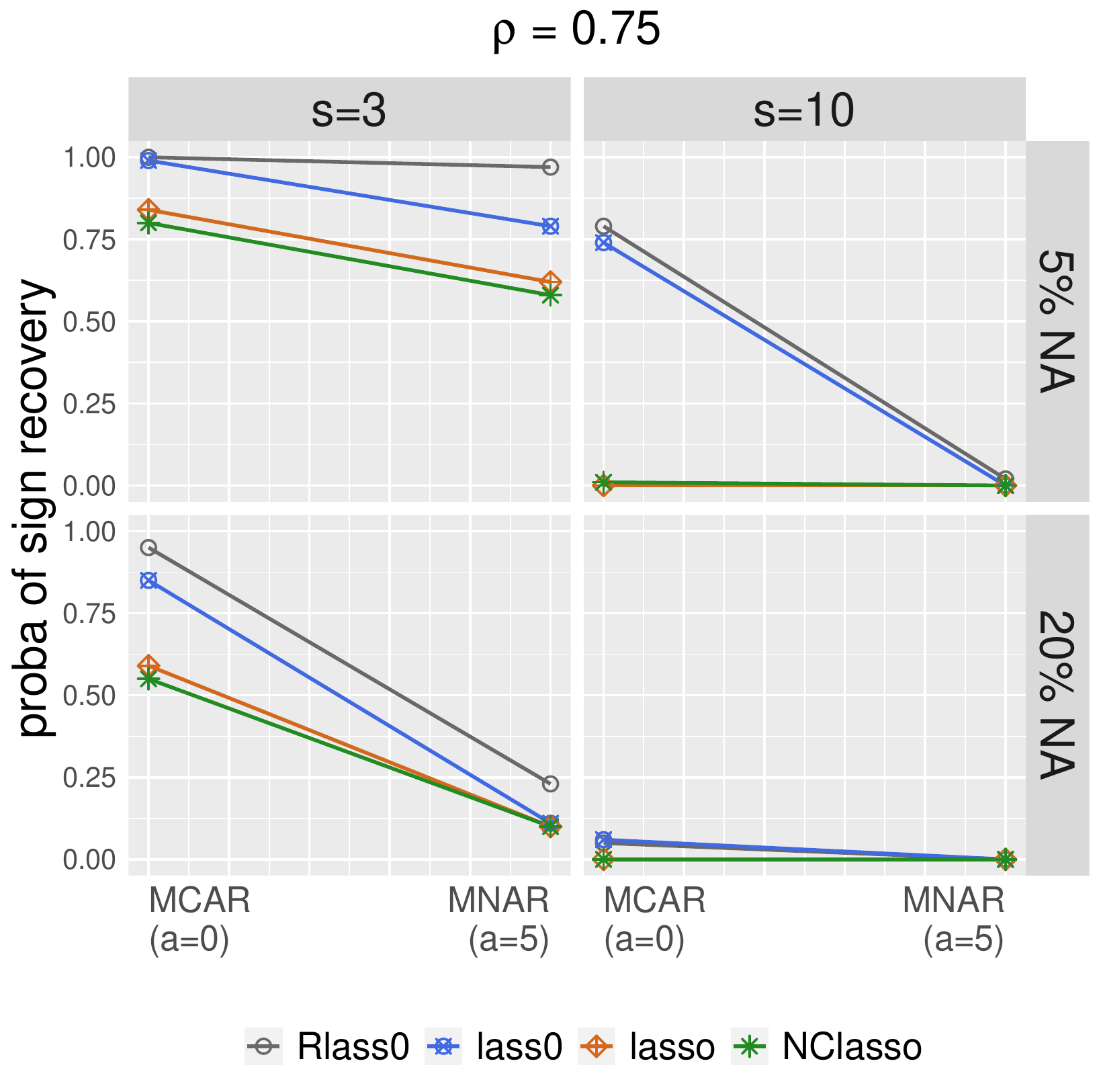} \\
(a) PSR in the non-correlated case & (b) PSR in the correlated case \\
& \\
\includegraphics[scale=0.5]{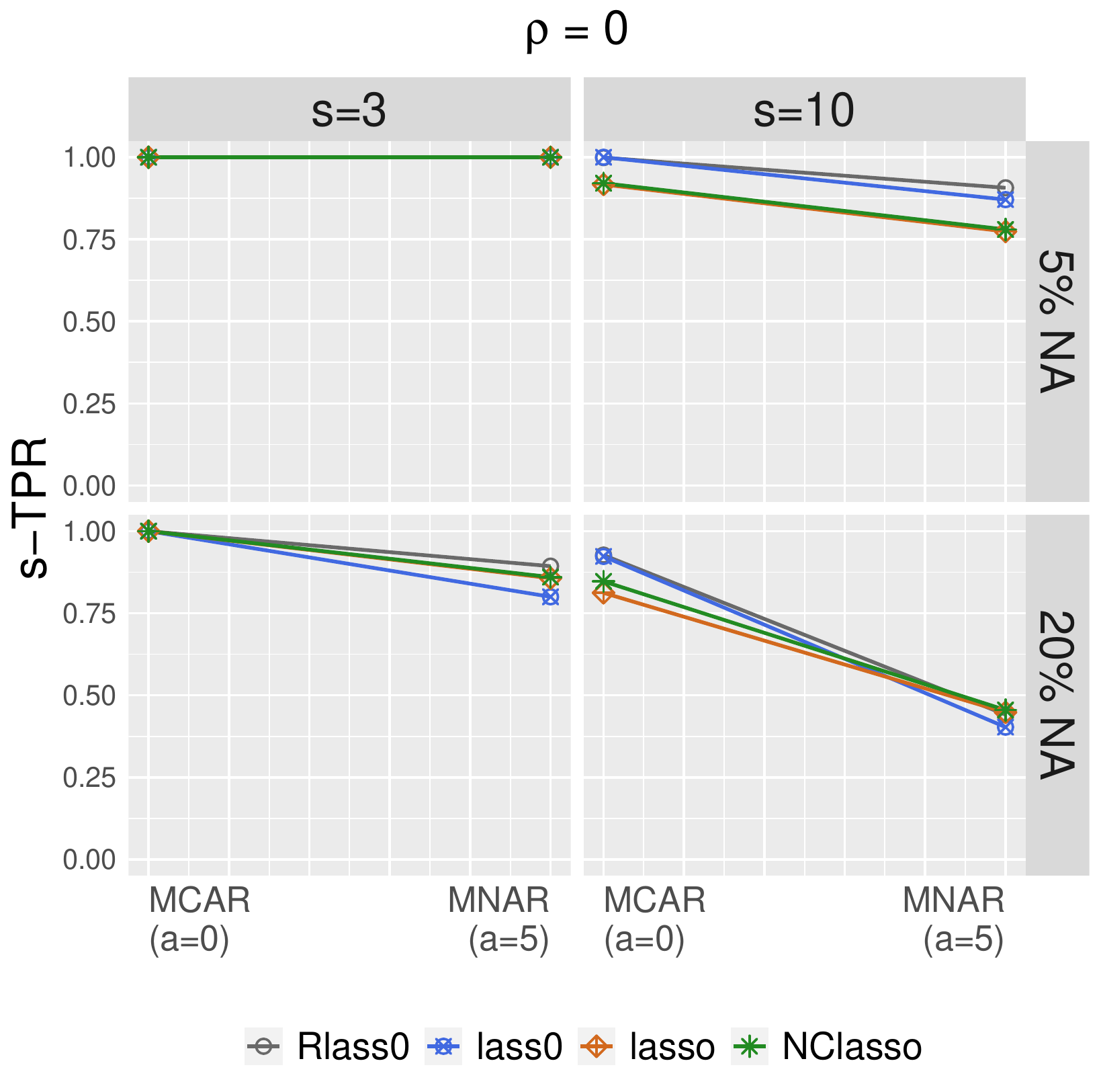} &
\includegraphics[scale=0.5]{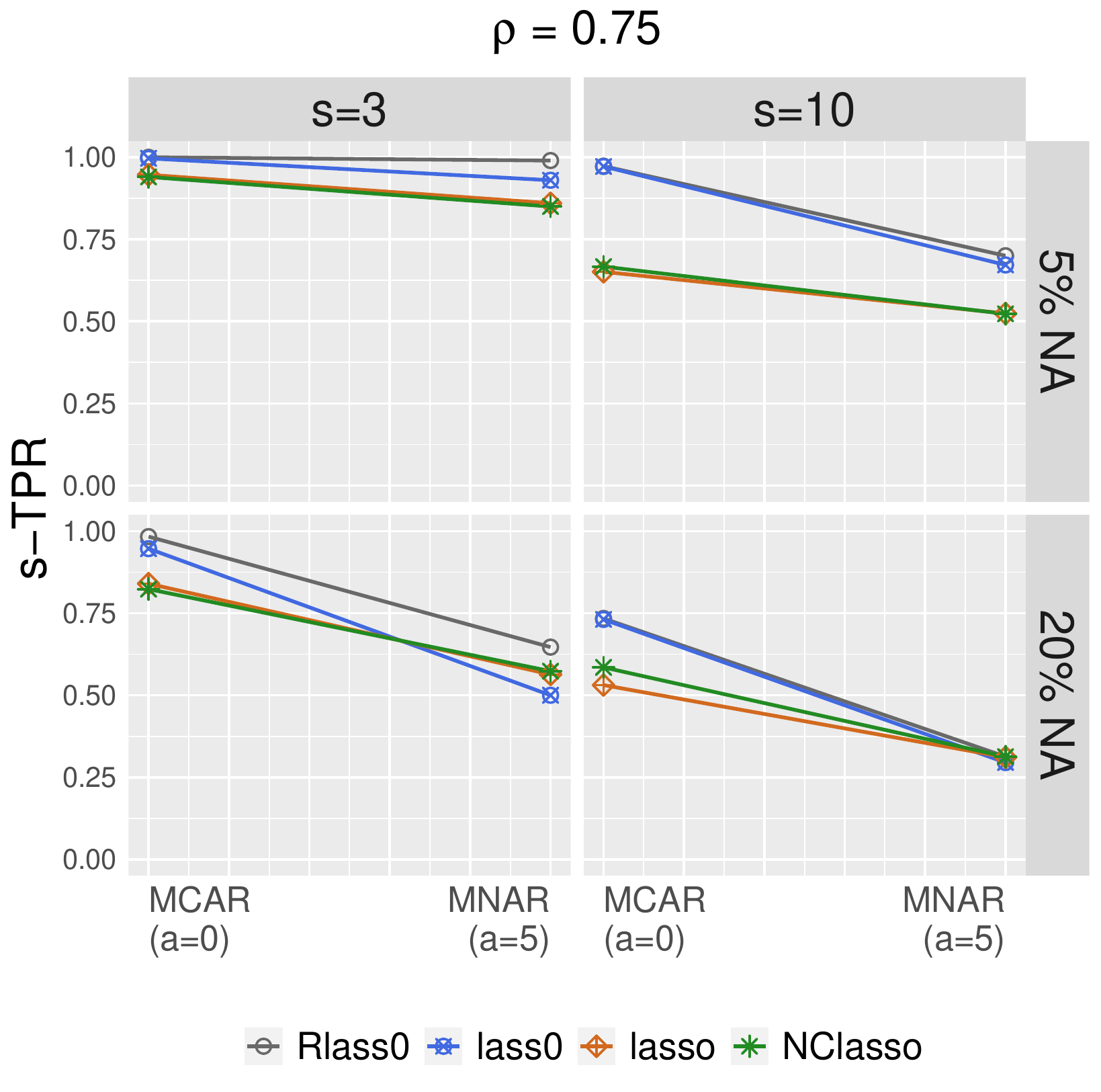} \\
(c) s-TPR in the non-correlated case & (d) s-TPR in the correlated case \\
\end{tabular}
\caption{\label{fig:soracle} PSR and s-TPR with an $s$-oracle tuning,  for sparsity levels , $s=3$ and $s=10$ (subplots columns), proportions of missing values $5\%$ or $20\%$ (subplots rows), and two missing data mechanisms (MCAR vs MNAR). }
\end{figure}

\paragraph{Increasing missingness -- High sparsity  ($20\%$ of NA and $s=3$).}
The benefit of Rlass0  is noticeable when increasing the percentage of missing data to $20\%$, for both performance indicators. Indeed, with no correlation (Figure \ref{fig:soracle} (a)(c)(bottom left)),  the improvement is clear when dealing with MNAR. With correlation (Figure \ref{fig:soracle} (b)(d)(bottom left)), Rlass0 outperforms the other methods: while the improvement can be marginal when compared to lass0 for MCAR, it becomes significant for MNAR. 

\paragraph{Lower sparsity ($s=10$).}
The performance of all estimators tends to deteriorate. One can identify two groups of estimators:  Rlass0 and  lass0 generally outperforms lasso and NClasso, except with a high proportion ($20\%$) of MNAR missing data  for which they all behave the same.  While comparable when $s=10$,  Rlass0 proves to be better than lass0 in the case of a small proportion of MNAR missing data ($5\%$).

\subsubsection{With automatic hyperparameter tuning} \label{section:automatic_tuning}

Figures~\ref{fig:PSR_auto} and \ref{fig:FDR_TPR_auto} point to the poor performance of lasso in terms of PSR for all experimental settings. 
The automatic tuning, being done by cross-validation, is known to lead to support overestimation. Indeed, its very good performance in sTPR is made at the cost of a very high sFDR.

\paragraph{Small missingness -- High sparsity ($5\%$ of NA and $s=3$).}
In Figures \ref{fig:PSR_auto}(a)(top left) and \ref{fig:FDR_TPR_auto}(a)(c)(top left), for the non-correlated case,  Rlass0, lass0 and ABSlope performs very well, providing a PSR and s-TPR of one, and a s-FDR of zero, either when dealing with MCAR or MNAR data (the lasso being already out of the game). In Figures \ref{fig:PSR_auto}(b)(top left) and \ref{fig:FDR_TPR_auto}(b)(d)(top left), adding correlation in the design matrix seems beneficial for ABSlope, at the price of high FDR, however. %\cb{pourquoi?} \pd{peut-\^etre car ABSLOPE estime (implicitement) la structure de corr\'elation?}
%\jj{j'enleverai surprisingly, c'est surprisingly pour abslope ou toi? ABslope est meilleur par construction avec correlation que non...}

\paragraph{Increasing missingness -- High sparsity ($20\%$ of NA and $s=3$).} 
With no correlation, one sees in Figure \ref{fig:PSR_auto}(a)(bottom left) that Rlass0 provides the best PSR, whatever the type of missing data is. One could also note that the performances in terms of PSR of either  lass0 or ABSLOPE are extremely variable depending on the type of missing data (MCAR or MNAR) considered: the PSR of lass0 is comparable to the one of Rlass0 when facing MCAR data and is much lower than the one of Rlass0 when facing MNAR data; the converse is true for ABSLOPE.  

Regarding the s-TPR and s-FDR results  in Figure \ref{fig:FDR_TPR_auto} (a-d)(bottom left), the following observations hold  in both correlated or non-correlated cases: 
\begin{enumerate}[(i)]
\item With MCAR data, all the methods behave similarly in terms of s-TPR, identifying correctly signs and coefficient locations in the support of~$\beta^0$, see Figure \ref{fig:FDR_TPR_auto}(a)(b)(bottom left);
\item With MNAR data, lasso and ABSLOPE remain stable in terms of s-TPR, providing an s-TPR of one, whereas the s-TPR of Rlass0 deteriorates (to 0.6 and 0.5 respectively for the non-correlated and correlated cases), and even worse for lass0, see Figure \ref{fig:FDR_TPR_auto}(a)(b)(bottom left);
\item Lasso and ABSLOPE lead to high s-FDR, % (above 0.25 with/out no correlation),
while lass0 and Rlass0 always give the best s-FDR, %($0$ and $0.15$ respectively without and with design correlation)
see Figure \ref{fig:FDR_TPR_auto}(c)(d)(bottom left).
\end{enumerate}

\paragraph{Lower sparsity ($s=10$).}
For  low missingness (5\%), see Figure \ref{fig:PSR_auto} (a)(b) (top right), ABSLOPE gives high PSR. 
In terms of s-TPR, lasso and ABSLOPE have high TPR. Moreover Rlass0 improves s-TPR compared to lass0 specially for a small proportion of MNAR missing data.
In terms of s-FDR, lass0 and Rlass0 bring very low s-FDR, proving their FDR stability with respect to MCAR/MNAR data, and correlation.

\begin{figure}
\hspace{-2cm}
\begin{tabular}{cc}
\includegraphics[scale=0.5]{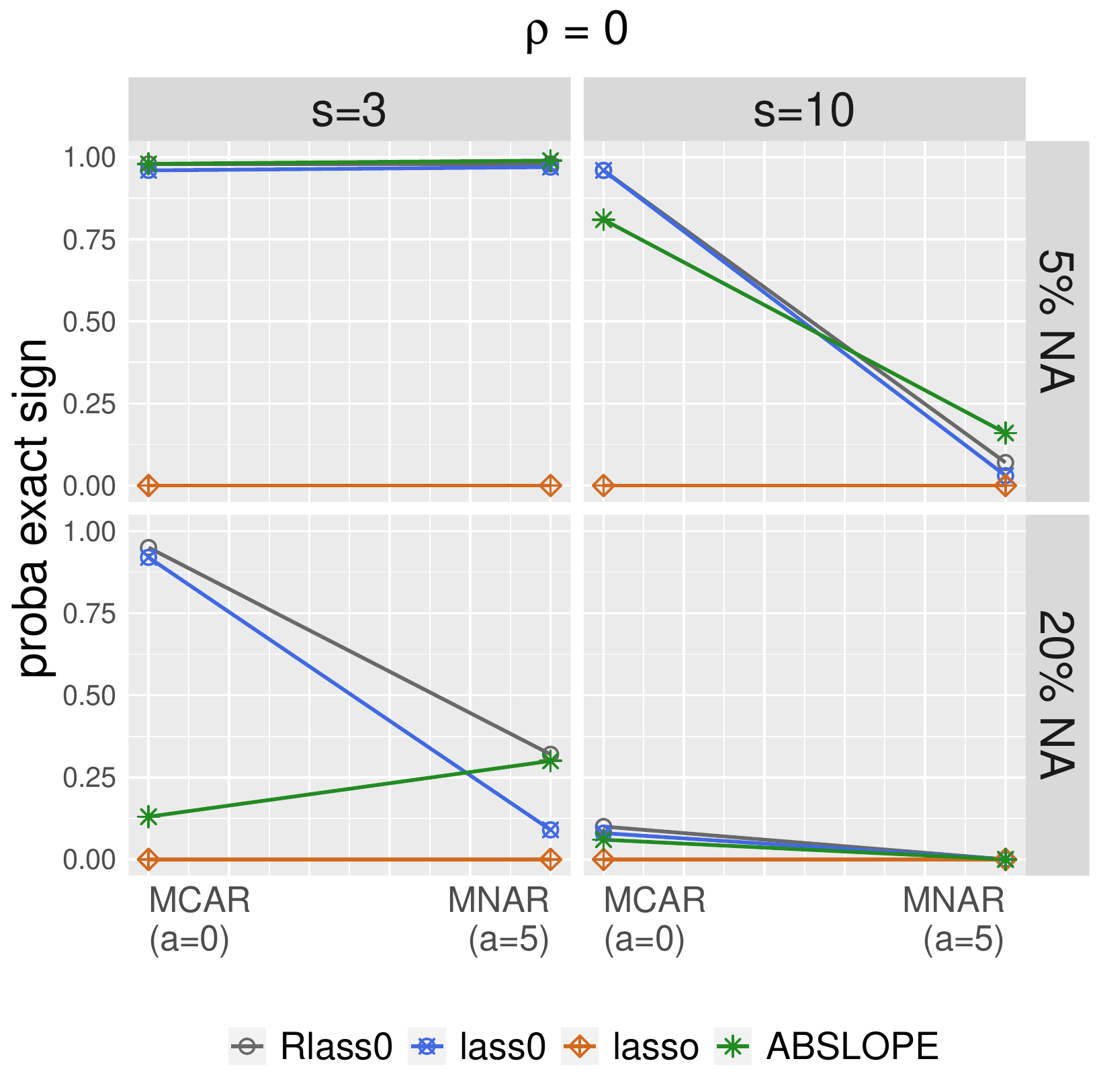} &
\includegraphics[scale=0.5]{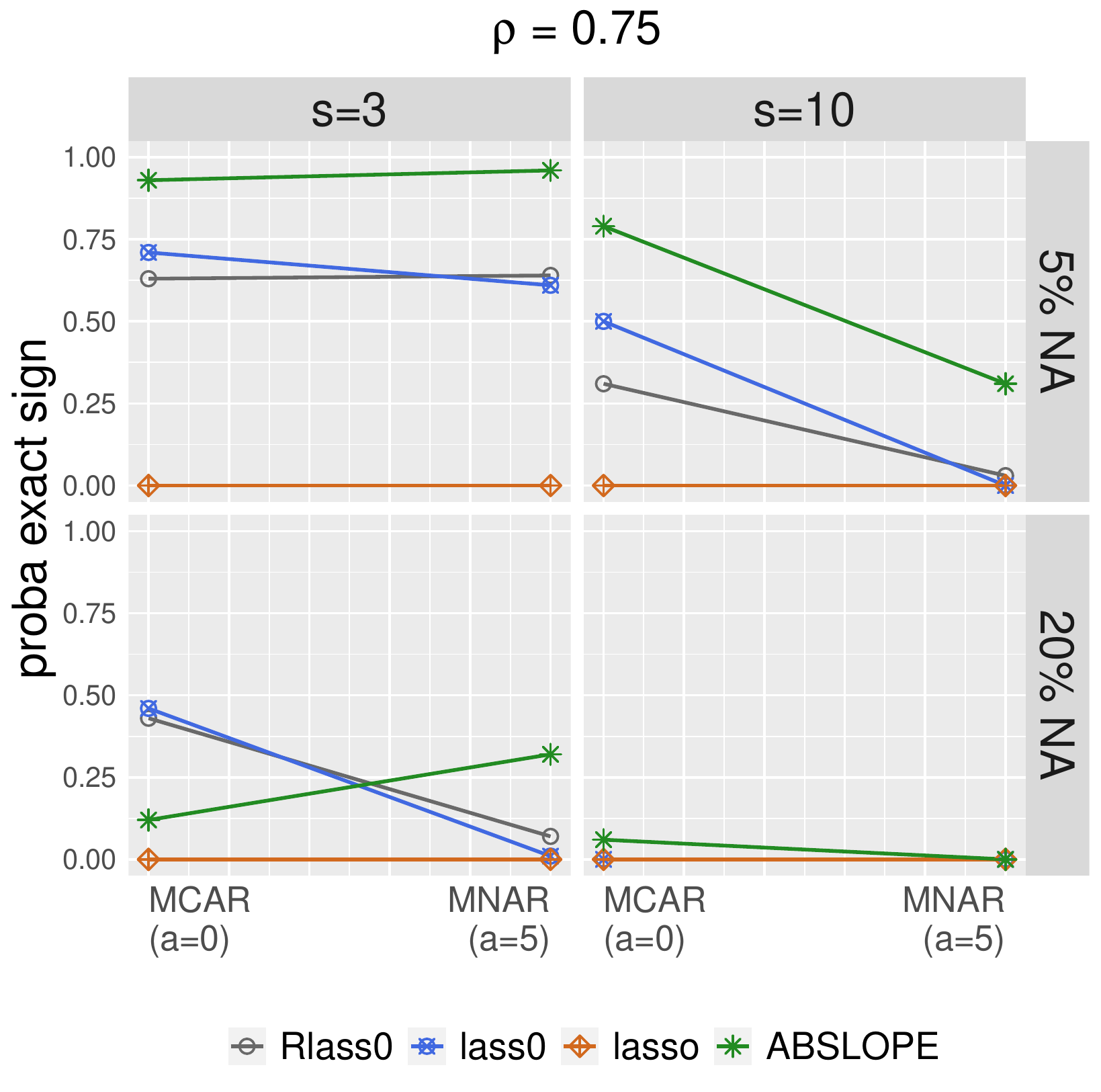} \\
(a) PSR in the non-correlated case & (b) PSR in the correlated case 
\end{tabular}
\caption{\label{fig:PSR_auto} PSR  with automatic tuning,  for sparsity levels $s=3$ and $s=10$ (subplots columns), proportions of missing values $5\%$ or $20\%$ (subplots rows), and two missing data mechanisms (MCAR vs MNAR).}
\end{figure}

\begin{figure}
\vspace{-2cm}
\hspace{-2cm}
\begin{tabular}{cc}
\includegraphics[scale=0.5]{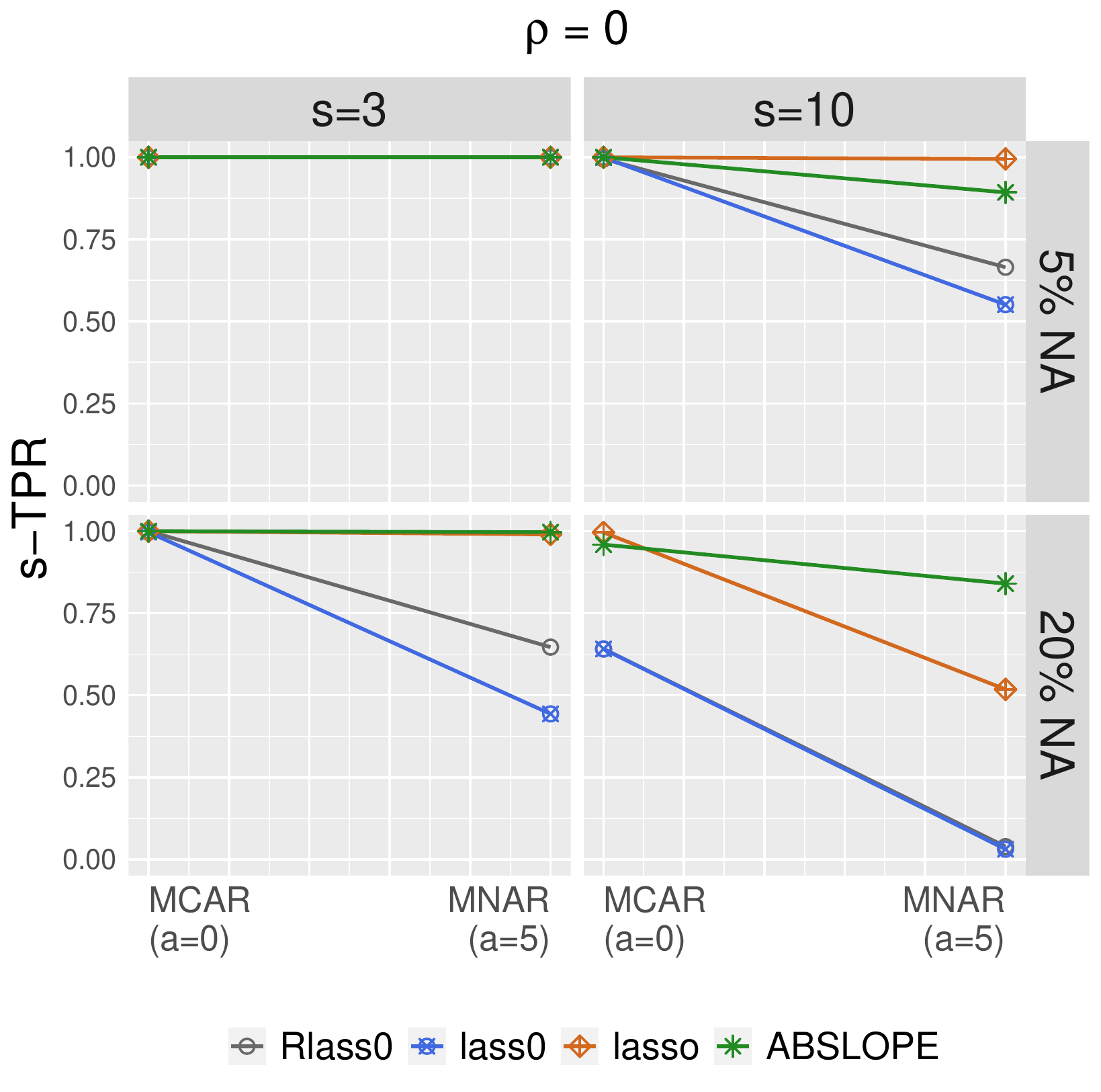} &
\includegraphics[scale=0.5]{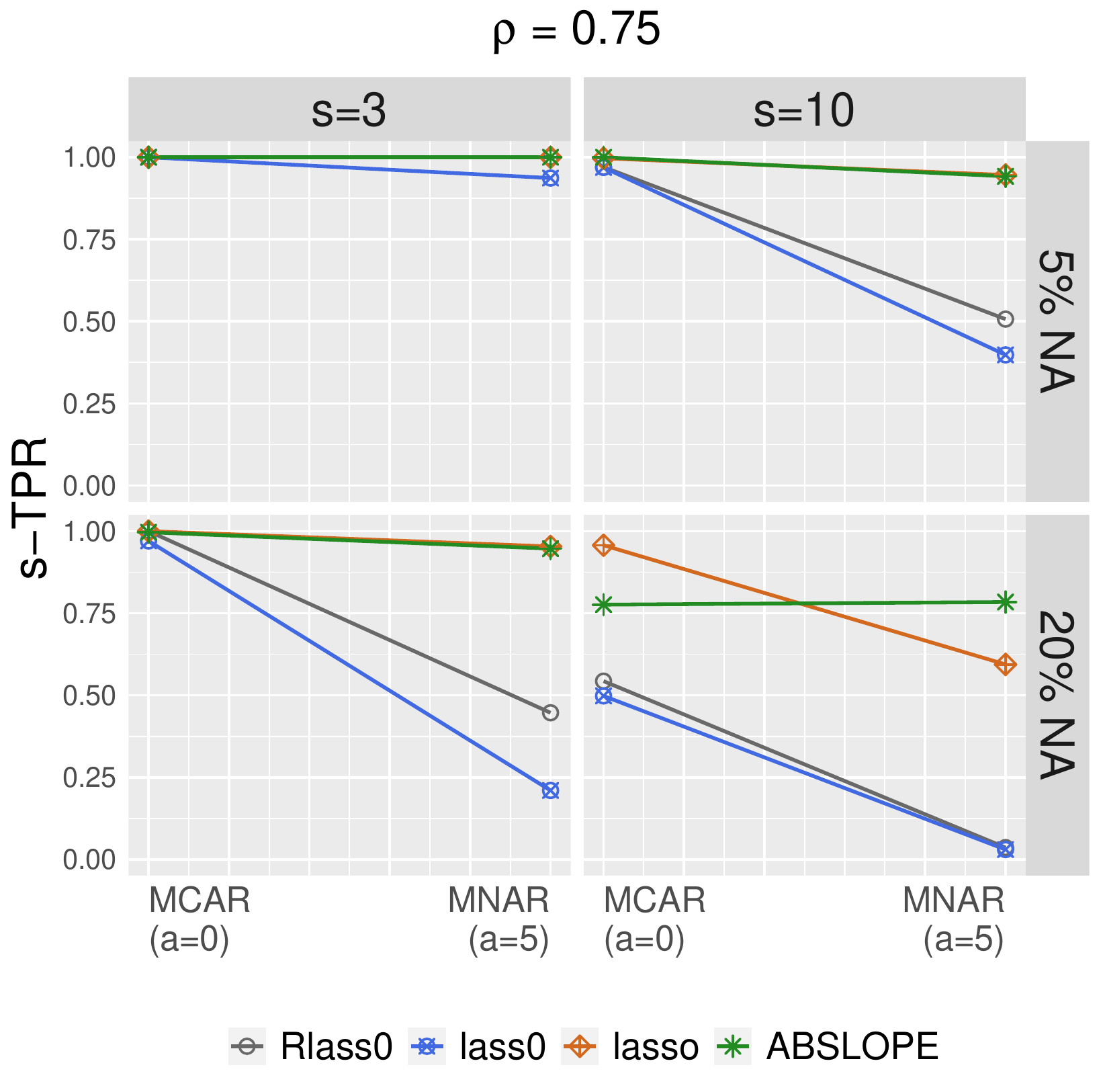} \\
(a) s-TPR in the non-correlated case & (b) s-TPR in the correlated case \\
& \\
\includegraphics[scale=0.5]{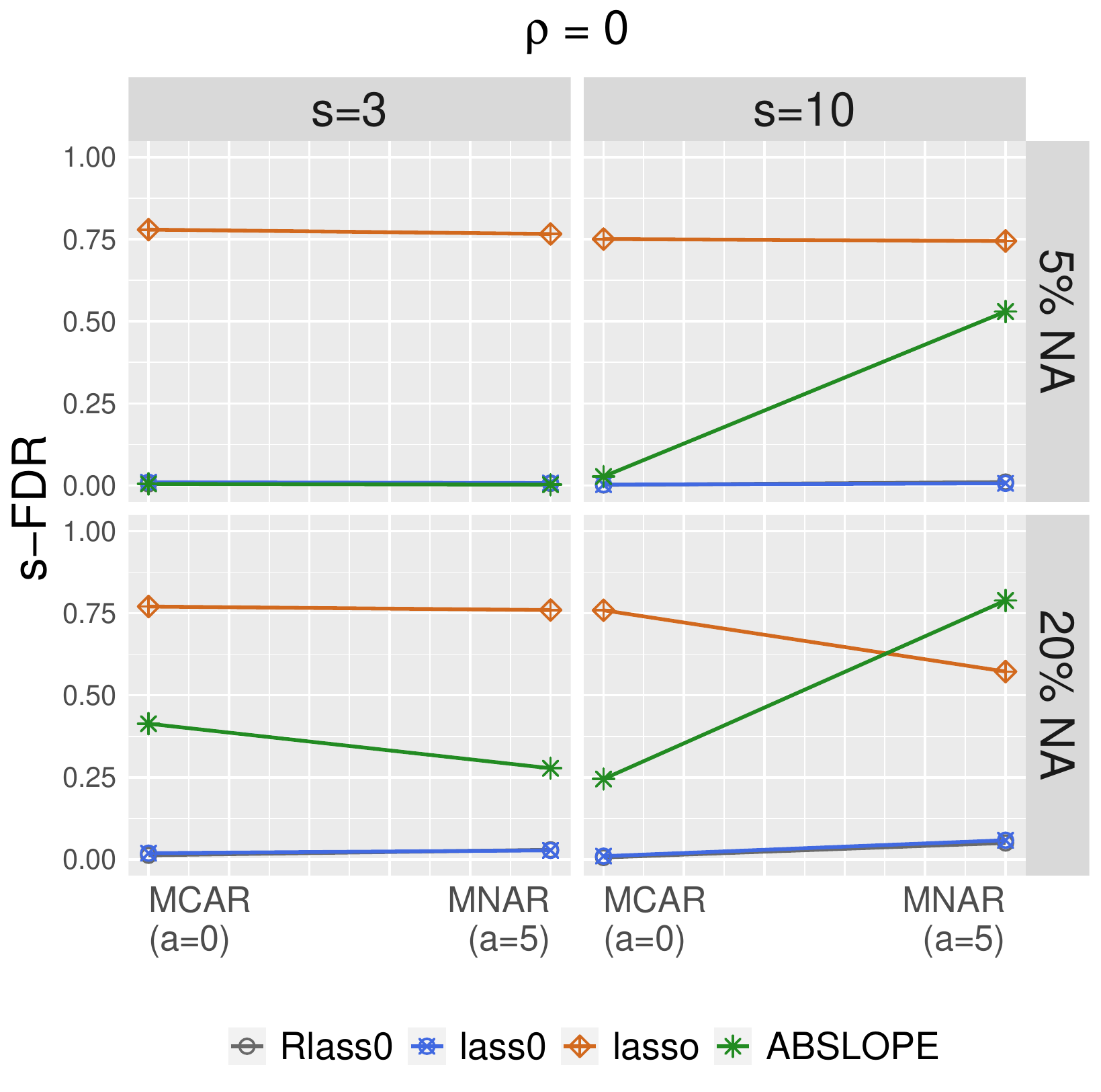} &
\includegraphics[scale=0.5]{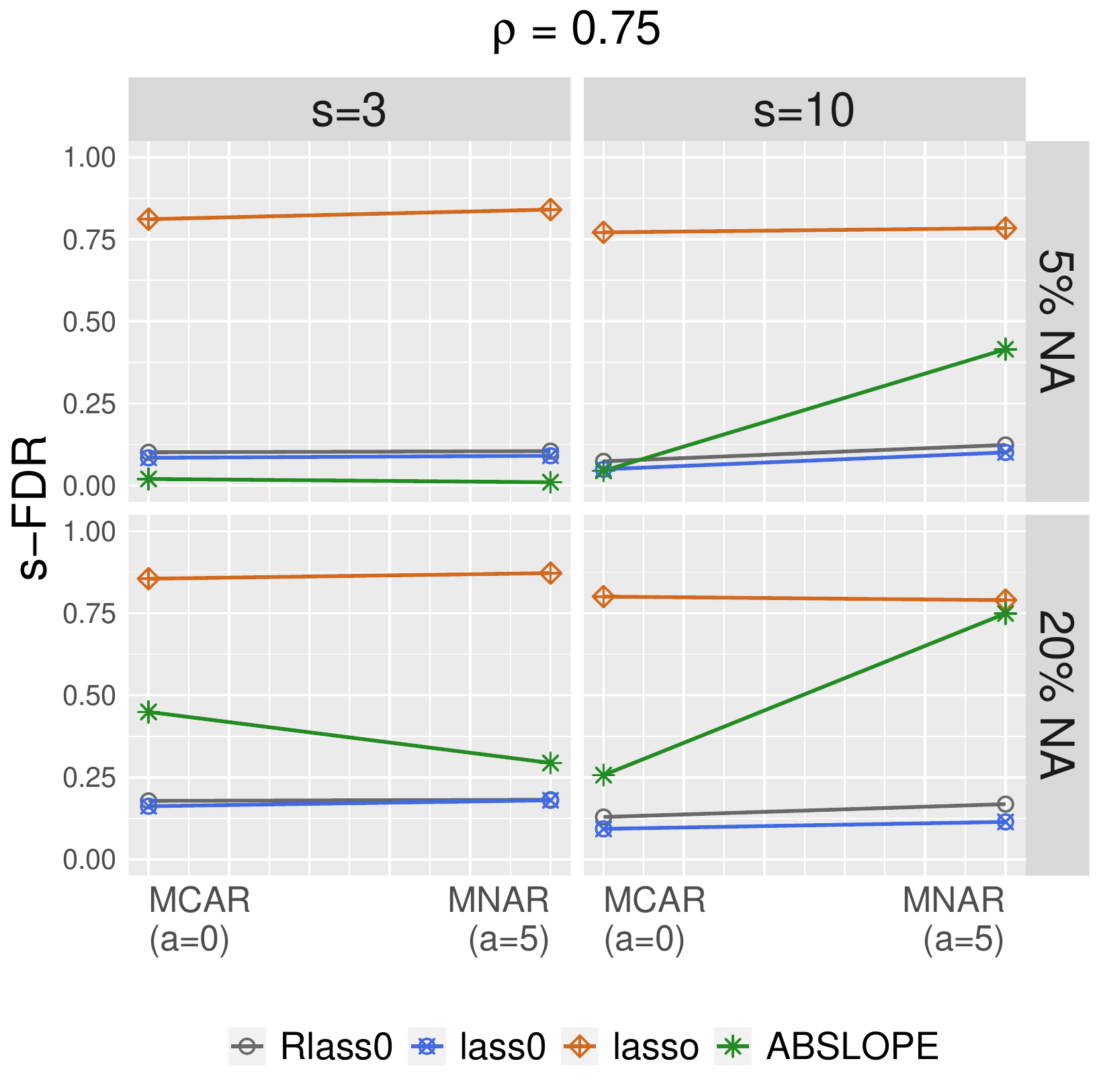} \\
(c) s-FDR in the non-correlated case & (d) s-FDR in the correlated case \\
\end{tabular}
\caption{\label{fig:FDR_TPR_auto} s-FDR and s-TPR with automatic tuning,  for sparsity levels $s=3$ and $s=10$ (subplots columns), proportions of missing values $5\%$ or $20\%$ (subplots rows), and two missing data mechanisms (MCAR vs MNAR).}
\end{figure}

\subsubsection{Summary and discussion}
The results of experiments with $s$-oracle tuning (Section~\ref{section:soracle_tuning}) show that the Robust Lasso-Zero performs better than competitors for sign recovery, and is more robust to MNAR data compared to its nonrobust counterpart when the sparsity index and/or proportion of missing entries is low. In particular, the Robust Lasso-Zero performs better than NClasso, one of the rare existing $\ell_1$-estimator designed to handle missing values.

While not designed to handle MNAR data, ABSLOPE appears to be a valid competitor in terms of s-TPR or PSR when the model complexity increases, and when dealing with MNAR data. Its poor performance in FDR in such settings reveals its tendency to overestimate the support of~$\beta^0$, under higher sparsity degrees, and  with informative MNAR missing data.

With automatic tuning (Section~\ref{section:automatic_tuning}), Robust Lasso-Zero is the best method overall.
Moreover, our results show that the choice of Robust Lasso-Zero tuned by QUT, with its low s-FDR, is particularly appropriate in cases where one wants to maintain a low proportion of false discoveries.

\section{Application to the Traumabase dataset} \label{sct:appli}

We illustrate our approach on the public health APHP (Assistance Publique Hopitaux de Paris)  TraumaBase$^{\mbox{\normalsize{\textregistered}}}$ Group for traumatized patients. Effective and timely management of trauma is crucial to improve outcomes, as delays or errors entail high risks for the patient.

\begin{figure}
	\includegraphics[scale=0.75]{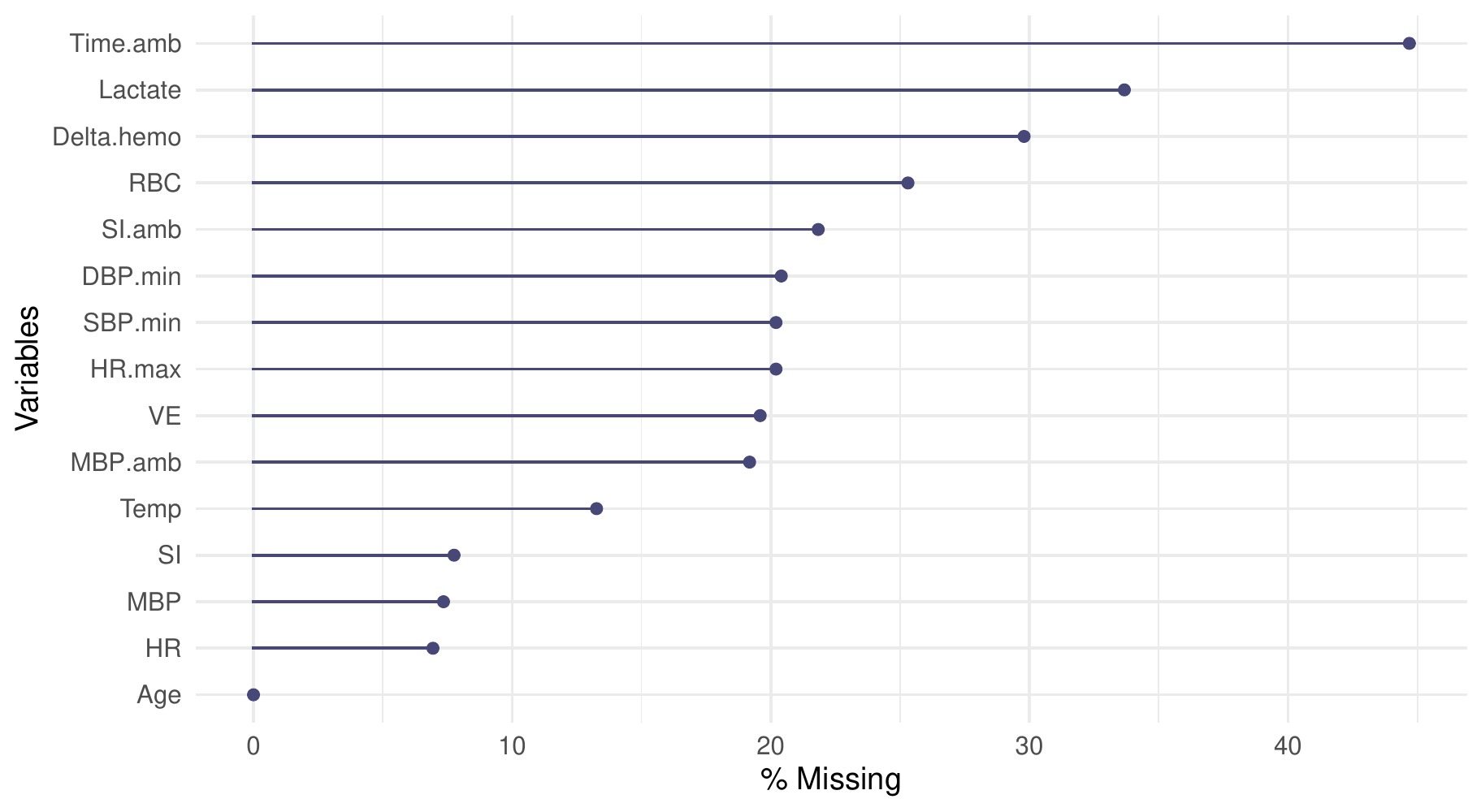}
	\caption{\label{fig:missingdataperc} Percentage of missing values in the Traumbase dataset.}
\end{figure}

\begin{table}
	\centering
	\begin{tabular}{cccccr}
		Variable & Rlass0 & lass0 & lasso & ABSLOPE \\
		\hline
		Age  & $-$  & 0 & $-$ & $-$ \\
		SI  & 0 & 0 & 0 & $-$ \\
		Delta.hemo  & 0 & 0 & 0 & $+$ \\
		Lactates & 0 & 0 & 0 & $+$ \\
		Temperature & 0 & 0 & 0 & $+$ \\
		VE         &  $-$& 0 & $-$ & 0 \\
		RBC        & $-$ & 0 & 0 & $-$ \\
		DBP.min & 0 & 0 & $-$ & $+$ \\
		HR.max & 0 & 0 & $-$ & 0 \\
		SI.amb & 0 & 0 & 0 & $+$ \\
	\end{tabular}
	\caption{\label{tab:signcoef} Sign of estimated effects on the platelet for Rlass0, lass0, lasso or ABSLOPE.  Variables not shown here are not selected by any  method.}
\end{table}

In our analysis, we focuse on one specific challenge: selecting a sparse model  from data containing missing covariates in order to explain the level of platelet. This model can aid creating an innovative response to the public health challenge of major trauma.
Explanatory variables for the level of platelet consist in fifteen quantitative variables containing missing values, which have been selected by doctors. They give clinical measurements on 490 patients. In Figure \ref{fig:missingdataperc}, one sees the percentage of missing values in each variable, varying from 0 to 45\% and  leading to 20\% is the whole dataset. Based on discussions with doctors, some variables may have informative missingness (M(N)AR variables). Both percentage and nature of missing data demonstrate the importance of taking appropriate account of missing data. More information  can be found in Appendix \ref{sec:variablestraumadataset}.

We compare the Robust Lasso-Zero to the Lasso-Zero, Lasso and ABSLOPE estimators. The signs of the coefficients are shown in Table~\ref{tab:signcoef}. Lass0 does not select any variable, whereas its robust counterpart selects three. 
According to doctors, the results given by the Robust Lasso-Zero methodology are the most coherent. Indeed, a negative effect of age (\textit{Age}), vascular filling (\textit{VE}) and blood transfusion (\textit{RBC}) was expected, as they all result in low platelet levels and therefore a higher risk of severe bleeding. Lasso similarly selects \textit{Age} and \textit{VE}, but also minimum value of diastolic blood pressure \textit{DBP.min} and the maximum heart rate \textit{HR.max}. The effect of \textit{DBP.min} is not what doctors expected. 
For ABSLOPE, the effects on platelets of delta Hemocue (\textit{Delta.Hemocue}), the lactates (\textit{Lactates}), the temperature (\textit{Temperature}) and the shock index measured on ambulance (\textit{SI.amb}), at odds with the effect of the shock index at hospital (\textit{SI}), are not in agreement with the doctors opinion either.

\section{Funding}

Aude Sportisse was supported by the French government, through the 3IA Côte d’Azur Investments in the Future project managed by the National Research Agency (ANR) with the reference number ANR-19-P3IA-0002.

\bibliographystyle{plainnat}
\bibliography{lass0NA}

\paragraph{Corresponding author} The corresponding author is Aude Sportisse (email: aude.sportisse@inria.fr). 

\appendix

\section{Proof of Theorem~\ref{thm:ext_tardivel}}
\label{app:proof_thm_ext_tardivel}
Lemma~\ref{lemma:identifiability_JP} implies that under the sign invariance assumption~(\ref{ass2:sign_invariance}), identifiability of $(\beta^{(r)}, \omega^{(r)})$ is equivalent to identifiability of $(\theta, \tilde{\theta}).$

\begin{proof}[Proof of Lemma~\ref{lemma:identifiability_JP}]
Note that $(\betahatJP_{\lambda}, \omegahatJP_{\lambda})$ is a solution to JP~(\ref{opt:JP_newscaling}) if and only if
$
(\betahatJP_{\lambda}, \omegahatJP_{\lambda}) = (\tilde{\beta}, \lambda^{-1} \tilde{\omega}),
$
where $(\tilde{\beta}, \tilde{\omega})$ is a solution to
\begin{equation}\label{BP_JP}
% \begin{aligned}
% &\min_{(\beta, \omega) \in \R^p \times \R^n} && \norm{\beta}_1 + \norm{\omega}_1 \\
% & \quad \text{ s.t.} && y = X\beta + \sqrt{n} \lambda^{-1} \omega.
\min_{(\beta,\ \omega) \in \R^p \times \R^n}  \norm{\beta}_1 + \norm{\omega}_1 \quad \text{s.t.} \quad  y = X\beta + \sqrt{n} \lambda^{-1} \omega.
% \end{aligned}
\end{equation}
So $(\beta^0, \omega^0)$ is identifiable with respect to $X$ and $\lambda > 0$ if and only if the pair $(\beta^0, \lambda \omega^0)$ is the unique solution of~(\ref{BP_JP}) when $y = X \beta^0 + \sqrt{n} \omega^0.$ But~(\ref{BP_JP}) is just Basis Pursuit with response vector $y \in \R^n$ and augmented matrix $\begin{bmatrix} X & \sqrt{n} \lambda^{-1} I_n \end{bmatrix},$ so by a result of \citet{daubechies2010} this is the case if and only if for every $(\beta, \omega) \neq (0, 0)$ such that $X \beta + \sqrt{n} \lambda^{-1} \omega = 0,$ we have
%
%$\begin{bmatrix} \beta \\ \omega \end{bmatrix} \in \ker{\begin{bmatrix} X & \sqrt{n} \lambda^{-1} I_n \end{bmatrix}} \setminus \left\{\begin{bmatrix} 0 \\ 0 \end{bmatrix}\right\},$
$
\abs{\sign(\beta^0)^T \beta + \sign(\omega^0)^T \omega} < \norm{\beta_{\Sbar}}_1 + \norm{\omega_{\Tbar}}_1,
$
which proves our statement.
\end{proof}

We will need the following auxiliary lemma.
\begin{lemma} \label{lemma:auxiliary_tardivel}
Under assumptions~(\ref{ass2:sign_invariance}) and~(\ref{ass2:inf_betamin}), if the pair $(\theta, \tilde{\theta})$ is identifiable with respect to $X$ and $\lambda,$ then for any $\epsilon \in \R^n,$
\[
\lim_{r \to + \infty} \frac{1}{u_r} \begin{bmatrix} \betahat^{\operatorname{JP}(r)}_{\lambda} - \beta^{(r)} \\
\omegahat^{\operatorname{JP}(r)}_{\lambda} - \omega^{(r)} \end{bmatrix} = \begin{bmatrix} 0 \\ 0 \end{bmatrix},
\]
where $u_r := \norm{\beta^{(r)}}_1 + \lambda \norm{\omega^{(r)}}_1.$
\end{lemma}
\begin{proof}
First note that by assumption~(\ref{ass2:inf_betamin}), $\lim_{r \to + \infty} u_r = + \infty.$ Now let $\epsilon \in \R^n$ and denote by $(\betahatJP_{\lambda}(\epsilon), \omegahatJP_{\lambda}(\epsilon))$ the JP solution when $y = \epsilon.$ In particular, one has $\epsilon = X \betahatJP_{\lambda}(\epsilon) + \sqrt{n} \omegahatJP_{\lambda}(\epsilon),$ so for every $r \in \N^*$, {let us consider $y^{(r)}$ as follows}
\[
y^{(r)} = X(\beta^{(r)} + \betahatJP_{\lambda}(\epsilon)) + \sqrt{n} (\omega^{(r)} + \omegahatJP_{\lambda}(\epsilon)).
\] 
Hence $(\beta^{(r)} +  \betahatJP_{\lambda}(\epsilon), \omega^{(r)} + \omegahatJP_{\lambda}(\epsilon))$ is feasible for JP when $y = y^{(r)},$ so
\begin{equation} \label{ineq:JPr}
\begin{aligned}
&\frac{\norm{\betahat^{\operatorname{JP}(r)}_{\lambda}}_1 + \lambda \norm{\omegahat^{\operatorname{JP}(r)}_{\lambda}}_1}{u_r}
\\
& \qquad \qquad \leq
\frac{\norm{\beta^{(r)} +  \betahatJP_{\lambda}(\epsilon)}_1 + \lambda \norm{\omega^{(r)} + \omegahatJP_{\lambda}(\epsilon)}_1}{u_r}
\\
& \qquad \qquad \leq
\frac{(\norm{\beta^{(r)}}_1 + \lambda \norm{\omega^{(r)}}_1) + (\norm{\betahatJP_{\lambda}(\epsilon)}_1 + \lambda \norm{\omegahatJP_{\lambda}(\epsilon)}_1)}{u_r} 
\\
& \qquad \qquad =
1 + \frac{\norm{\betahatJP_{\lambda}(\epsilon)}_1 + \lambda \norm{\omegahatJP_{\lambda}(\epsilon)}_1}{u_r}.
\end{aligned}
\end{equation}
Therefore
\begin{equation} \label{ineq:diff}
\begin{aligned}
&\frac{1}{u_r} (\norm{\betahat^{\operatorname{JP}(r)}_{\lambda} - \beta^{(r)}}_1 + \lambda \norm{\omegahat^{\operatorname{JP}(r)}_{\lambda} - \omega^{(r)}}_1)
\\
& \qquad \qquad \leq
\frac{1}{u_r}((\norm{\beta^{(r)}}_1 + \lambda \norm{\omega^{(r)}}_1) + (\norm{\betahat^{\operatorname{JP}(r)}_{\lambda}}_1 + \lambda \norm{\omegahat^{\operatorname{JP}(r)}_{\lambda}}_1))
\\
&\qquad \qquad =
1 + \frac{\norm{\betahat^{\operatorname{JP}(r)}_{\lambda}}_1 + \lambda \norm{\omegahat^{\operatorname{JP}(r)}_{\lambda}}_1}{u_r}
\\
&\qquad \qquad \leq
2 +  \frac{\norm{\betahatJP_{\lambda}(\epsilon)}_1 + \lambda \norm{\omegahatJP_{\lambda}(\epsilon)}_1}{u_r},
\end{aligned}
\end{equation}
using~(\ref{ineq:JPr}) for last inequality. Since $\lim_{r \to + \infty} \frac{\norm{\betahatJP_{\lambda}(\epsilon)}_1 + \lambda \norm{\omegahatJP_{\lambda}(\epsilon)}_1}{u_r} = 0,$ and since $\begin{bmatrix}\beta \\ \omega \end{bmatrix} \mapsto \norm{\beta}_1 + \lambda \norm{\omega}_1$ defines a norm on $\R^{p + n},$ one deduces that the sequence $\frac{1}{u_r} \begin{bmatrix} \betahat^{\operatorname{JP}(r)}_{\lambda} - \beta^{(r)} \\
\omegahat^{\operatorname{JP}(r)}_{\lambda} - \omega^{(r)} \end{bmatrix}$ is bounded. Therefore we need to check that every convergent subsequence converges to zero. Let 
\[
\frac{1}{u_{\phi(r)}} \begin{bmatrix} \betahat^{\operatorname{JP}(\phi(r))}_{\lambda} - \beta^{(\phi(r))} \\
\omegahat^{\operatorname{JP}(\phi(r))}_{\lambda} - \omega^{(\phi(r))} \end{bmatrix}
\] 
(with $\phi: \N^* \to \N^*$ strictly increasing) be an arbitrary convergent subsequence. Since 
\begin{equation} \label{eq:norm_r_1}
\frac{\norm{\beta^{(r)}}_1 + \lambda \norm{\omega^{(r)}}_1}{u_r} = 1
\end{equation} 
for every $r,$ and by~(\ref{ineq:JPr}), the sequences $\frac{1}{u_r} \begin{bmatrix}\beta^{(r)} \\ \omega^{(r)} \end{bmatrix}$ and $\frac{1}{u_r} \begin{bmatrix}\betahat^{\operatorname{JP}(r)}_{\lambda} \\ \omegahat^{\operatorname{JP}(r)}_{\lambda} \end{bmatrix}$ are bounded as well. Hence without loss of generality (otherwise, reduce the subsequence), 
\begin{equation} \label{limit1}
\lim_{r \to + \infty} \frac{1}{u_{\phi(r)}} \begin{bmatrix}\beta^{(\phi(r))} \\ \omega^{(\phi(r))} \end{bmatrix} = \begin{bmatrix} \nu_1 \\ \nu_2 \end{bmatrix},
\end{equation}
and
\begin{equation} \label{limit2}
\lim_{r \to +\infty} \frac{1}{u_{\phi(r)}} \begin{bmatrix} \betahat^{\operatorname{JP}(\phi(r))}_{\lambda} \\ \omegahat^{\operatorname{JP}(\phi(r))}_{\lambda} \end{bmatrix} = \begin{bmatrix} \nu_1' \\ \nu_2' \end{bmatrix}
\end{equation}
for some $\begin{bmatrix} \nu_1 \\ \nu_2 \end{bmatrix}, \begin{bmatrix} \nu_1' \\ \nu_2' \end{bmatrix} \in \R^{p+n}.$ By~(\ref{eq:norm_r_1}), one necessarily has
\begin{equation} \label{limit_norm}
\norm{\nu_1}_1 + \lambda \norm{\nu_2}_1 = 1,
\end{equation}
and~(\ref{ineq:JPr}) implies that
\begin{equation} \label{limit_norm_bound}
\norm{\nu_1'}_1 + \lambda \norm{\nu_2'}_1 \leq 1.
\end{equation}
Now
\begin{equation*}
\begin{aligned}
\lim_{r \to + \infty} \frac{X(\betahat^{\operatorname{JP}(r)}_{\lambda} - \beta^{(r)}) + \sqrt{n} (\omegahat^{\operatorname{JP}(r)}_{\lambda} - \omega^{(r)})}{u_r}
&=
\lim_{r \to + \infty} \frac{y^{(r)} - (X\beta^{(r)} + \sqrt{n} \omega^{(r)})}{u_r}
\\
&=
\lim_{r \to + \infty} \frac{\epsilon}{u_r} = 0,
\end{aligned}
\end{equation*}
so one deduces that 
\begin{equation*} 
\lim_{r \to + \infty} \begin{bmatrix} X & \sqrt{n} I_n \end{bmatrix} \begin{bmatrix} \betahat^{\operatorname{JP}(\phi(r))}_{\lambda} / u_{\phi(r)} \\ \omegahat^{\operatorname{JP}(\phi(r))}_{\lambda} / u_{\phi(r)} \end{bmatrix}
= 
\lim_{r \to + \infty} \begin{bmatrix} X & \sqrt{n} I_n \end{bmatrix} \begin{bmatrix} \beta^{(\phi(r))}/ u_{\phi(r)} \\ \omega^{(\phi(r))} / u_{\phi(r)} \end{bmatrix},
\end{equation*}
so by~(\ref{limit1}) and~(\ref{limit2}),
\begin{equation} \label{limit_signal}
\begin{bmatrix} X & \sqrt{n} I_n \end{bmatrix} \begin{bmatrix} \nu_1' \\ \nu_2' \end{bmatrix}
=
\begin{bmatrix} X & \sqrt{n} I_n \end{bmatrix} \begin{bmatrix} \nu_1 \\ \nu_2 \end{bmatrix}.
\end{equation}
Assuming for now that $(\nu_1, \nu_2)$ is identifiable with respect to $X$ and $\lambda,$ equality~(\ref{limit_signal}) together with~(\ref{limit_norm}) and~(\ref{limit_norm_bound}) imply that $\begin{bmatrix} \nu_1' \\ \nu_2' \end{bmatrix} =  \begin{bmatrix} \nu_1 \\ \nu_2 \end{bmatrix},$ hence
\begin{equation*}
\lim_{r \to + \infty} \frac{1}{u_{\phi(r)}} \begin{bmatrix} \betahat^{\operatorname{JP}(\phi(r))}_{\lambda} - \beta^{(\phi(r))} \\
\omegahat^{\operatorname{JP}(\phi(r))}_{\lambda} - \omega^{(\phi(r))} \end{bmatrix}
= 
\begin{bmatrix} \nu_1' \\ \nu_2' \end{bmatrix} -  \begin{bmatrix} \nu_1 \\ \nu_2 \end{bmatrix}
= \begin{bmatrix} 0 \\ 0 \end{bmatrix}.
\end{equation*}
It remains to check that $(\nu_1, \nu_2)$ is identifiable with respect to $X$ and $\lambda,$ which we will do using Lemma~\ref{lemma:identifiability_JP}. Note that~(\ref{limit1}) and assumption~(\ref{ass2:sign_invariance}) imply
\begin{align}
&\sign(\nu_1) = \theta - \theta', \label{sign_nu1}
\\
&\sign(\nu_2) = \tilde{\theta} - \tilde{\theta}', \label{sign_nu2}
\end{align}
where $\theta'_j := \theta_j \bm{1}_{\{\nu_{1,j} = 0, \theta_j \neq 0\}},$ and $\tilde{\theta}'_j = \tilde{\theta}_j \bm{1}_{\{\nu_{2,j} = 0, \tilde{\theta}_j \neq 0\}},$ and hence
\begin{align}
&\overline{\supp(\nu_1)} = \overline{\supp(\theta)} \sqcup \supp(\theta') = \Sbar  \sqcup \supp(\theta'), \label{disjoint_union_suppnu1}
\\
&\overline{\supp(\nu_2)} = \overline{\supp(\tilde{\theta})} \sqcup \supp{\tilde{\theta}'} = \Tbar \sqcup \supp{\tilde{\theta}'}. \label{disjoint_union_suppnu2}
\end{align}
Consider a pair $(\beta, \omega) \neq (0, 0)$ such that $X \beta + \sqrt{n} \lambda^{-1} \omega = 0.$ By~(\ref{sign_nu1}) and~(\ref{sign_nu2}),
\begin{equation} \label{ineq:lemma_indent}
\begin{aligned}
\abs{\sign(\nu_1)^T \beta + \sign(\nu_2)^T \omega}
&=
\abs{(\theta - \theta')^T \beta + (\tilde{\theta} - \tilde{\theta}')^T \omega}
\\
&\leq
\abs{\theta^T \beta + \tilde{\theta}^T \omega} + \abs{(\theta')^T \beta} + \abs{(\tilde{\theta}')^T \omega}.
\end{aligned}
\end{equation}
But since $(\theta, \tilde{\theta})$ is identifiable with respect to $X$ and $\lambda,$ Lemma~\ref{lemma:identifiability_JP} implies
%\begin{equation} \label{sign_identifiable}
$\abs{\theta^T \beta + \tilde{\theta}^T \omega} < \norm{\beta_{\Sbar}}_1 + \norm{\omega_{\Tbar}}_1.$
%\end{equation}
Plugging this into~(\ref{ineq:lemma_indent}) gives
\begin{equation*}
\begin{aligned}
\abs{\sign(\nu_1)^T \beta + \sign(\nu_2)^T \omega} 
&<
\norm{\beta_{\Sbar}}_1 + \norm{\omega_{\Tbar}}_1 +  \abs{(\theta')^T \beta} + \abs{(\tilde{\theta}')^T \omega}
\\
&\leq
\norm{\beta_{\Sbar}}_1 + \norm{\beta_{\supp(\theta')}}_1 + \norm{\omega_{\Tbar}}_1 + \norm{\omega_{\supp(\tilde{\theta}')}}_1
\\
&=
\norm{\beta_{\overline{\supp(\nu_1)}}}_1 + \norm{\omega_{\overline{\supp(\nu_2)}}}_1,
\end{aligned}
\end{equation*}
where the equality comes from~(\ref{disjoint_union_suppnu1}) and~(\ref{disjoint_union_suppnu2}). By Lemma~\ref{lemma:identifiability_JP}, one concludes that $(\nu_1, \nu_2)$ is identifiable with respect to $X$ and $\lambda.$
\end{proof}

\begin{proof}[Proof of Theorem~\ref{thm:ext_tardivel}]
Let us assume that $(\theta, \tilde{\theta})$ is identifiable with respect to $X$ and $\lambda,$ and let $\epsilon \in \R^n.$ By Lemma~\ref{lemma:auxiliary_tardivel},
\begin{equation} \label{consequence_auxiliary}
\lim_{r \to + \infty} \frac{1}{u_r} \begin{bmatrix} \betahat^{\operatorname{JP}(r)}_{\lambda} - \beta^{(r)} \\
\omegahat^{\operatorname{JP}(r)}_{\lambda} - \omega^{(r)} \end{bmatrix} = \begin{bmatrix} 0 \\ 0 \end{bmatrix}.
\end{equation}
Since
\[
\min \{1, \lambda \}\max\{\norm{\beta^{(r)}}_{\infty}, \norm{\omega^{(r)}}_{\infty}\}
\leq 
u_r
\leq
(\abs{S^0} + \lambda \abs{T^0}) \max\{\norm{\beta^{(r)}}_{\infty}, \norm{\omega^{(r)}}_{\infty}\},
\]
(\ref{consequence_auxiliary}) is equivalent to
$
\lim_{r \to + \infty} \frac{1}{\max\{\norm{\beta^{(r)}}_{\infty}, \norm{\omega^{(r)}}_{\infty}\}} \begin{bmatrix} \betahat^{\operatorname{JP}(r)}_{\lambda} - \beta^{(r)} \\
\omegahat^{\operatorname{JP}(r)}_{\lambda} - \omega^{(r)} \end{bmatrix} = \begin{bmatrix} 0 \\ 0 \end{bmatrix}.
$
Therefore there exists $R > 0$ such that for every $r \geq R,$
\begin{equation} \label{bound:error1}
\norm{\betahat^{\operatorname{JP}(r)}_{\lambda} - \beta^{(r)}}_{\infty}
<
\frac{q}{2} \max\{\norm{\beta^{(r)}}_{\infty}, \norm{\omega^{(r)}}_{\infty}\}
\end{equation}
and
\begin{equation} \label{bound:error2}
\norm{\omegahat^{\operatorname{JP}(r)}_{\lambda} - \omega^{(r)}}_{\infty}
<
\frac{q}{2} \max\{\norm{\beta^{(r)}}_{\infty}, \norm{\omega^{(r)}}_{\infty}\}.
\end{equation}
Setting $\tau := \frac{q}{2} \max\{\norm{\beta^{(r)}}_{\infty}, \norm{\omega^{(r)}}_{\infty}\},$~(\ref{bound:error1}) implies that $\abs{\betahat^{\operatorname{JP}(r)}_{\lambda, j}} < \tau$ for every $j \notin S^0,$ hence $\betahat^{\operatorname{TJP}(r)}_{(\lambda, \tau), j} = 0.$ If $j \in S^0,$ assumption~(\ref{ass2:bounded_ratio}) implies
\begin{equation} \label{beta_min_2tau}
\abs{\beta^{(r)}_j} \geq \beta^{(r)}_{\min} \geq 2 \tau,
\end{equation}
and by~(\ref{bound:error1}), we have
\begin{equation} \label{bound:error3}
\abs{\betahat^{\operatorname{JP}(r)}_{\lambda, j} - \beta^{(r)}_j} < \tau,
\end{equation}
so~(\ref{beta_min_2tau}) and~(\ref{bound:error3}) together imply $\abs{\betahat^{\operatorname{JP}(r)}_{\lambda, j} } > \tau$ and $\sign(\betahat^{\operatorname{JP}(r)}_{\lambda, j}) = \sign(\beta^{(r)}_j).$ So we conclude that $\sign(\betahat^{\operatorname{TJP}(r)}_{(\lambda, \tau)}) = \sign(\beta^{(r)}).$ Analogously,~(\ref{bound:error2}) implies $\sign(\omegahat^{\operatorname{TJP}(r)}_{(\lambda, \tau)}) = \sign(\omega^{(r)}).$

Conversely, let us assume that for some $\epsilon \in \R^n,$ $r \in \N^*$ and $\tau > 0,$
\begin{equation} \label{equal_signs}
\sign(\betahat^{\operatorname{TJP}(r)}_{(\lambda, \tau)}) = \theta, 
\quad \sign(\omegahat^{\operatorname{TJP}(r)}_{(\lambda, \tau)}) = \tilde{\theta}.
\end{equation}
Note that the JP solution $(\betahat^{\operatorname{JP}(r)}_{\lambda}, \omegahat^{\operatorname{JP}(r)}_{\lambda})$ is unique by assumption, hence $(\betahat^{\operatorname{JP}(r)}_{\lambda}, \omegahat^{\operatorname{JP}(r)}_{\lambda})$ is identifiable with respect to $X$ and $\lambda.$ Now by~(\ref{equal_signs}), all nonzero components of $\theta$ and $\tilde{\theta}$ must have the same sign as the corresponding entries of $\betahat^{\operatorname{JP}(r)}_{\lambda}$ and $\omegahat^{\operatorname{JP}(r)}_{\lambda}$ respectively. Hence
\begin{equation} \label{sign_decomp}
\begin{aligned}
&\theta = \sign(\theta) = \sign(\betahat^{\operatorname{JP}(r)}_{\lambda}) - \delta, \\
&\tilde{\theta} = \sign(\tilde{\theta}) = \sign(\omegahat^{\operatorname{JP}(r)}_{\lambda}) - \tilde{\delta},
\end{aligned}
\end{equation}
where $\delta_j = \sign(\betahat^{\operatorname{JP}(r)}_{\lambda, j}) \bm{1}_{\{\betahat^{\operatorname{JP}(r)}_{\lambda, j} \neq 0, \theta_j = 0\}}$ and $\tilde{\delta}_i = \sign(\omegahat^{\operatorname{JP}(r)}_{\lambda, i}) \bm{1}_{\{\omegahat^{\operatorname{JP}(r)}_{\lambda, i} \neq 0, \tilde{\theta}_i = 0\}},$ and
\begin{equation} \label{supp_decomp}
\begin{aligned}
&\Sbar = \overline{\supp(\theta)} = \overline{\supp(\betahat^{\operatorname{JP}(r)}_{\lambda})} \sqcup \supp(\delta)
\\
&\Tbar = \overline{\supp(\tilde{\theta})} = \overline{\supp(\omegahat^{\operatorname{JP}(r)}_{\lambda})} \sqcup \supp(\tilde{\delta}).
\end{aligned}
\end{equation}
In order to apply Lemma~\ref{lemma:identifiability_JP}, let us consider a pair $(\beta, \omega) \neq (0, 0)$ such that $X \beta + \sqrt{n} \lambda^{-1} \omega = 0.$  By~(\ref{sign_decomp}), one has
\begin{equation*}
\begin{aligned}
\abs{\theta^T \beta + \tilde{\theta}^T \omega}
&=
\abs{\sign(\betahat^{\operatorname{JP}(r)}_{\lambda})^T \beta  - \delta^T \beta + \sign(\omegahat^{\operatorname{JP}(r)}_{\lambda})^T \omega - \tilde{\delta}^T \omega}
\\
&\leq
\abs{\sign(\betahat^{\operatorname{JP}(r)}_{\lambda})^T \beta + \sign(\omegahat^{\operatorname{JP}(r)}_{\lambda})^T \omega} + \abs{\delta^T \beta} + \abs{\tilde{\delta}^T \omega}
\\
&\leq
\norm{\beta_{\overline{\supp(\betahat^{\operatorname{JP}(r)}_{\lambda})}}}_1 + \norm{\omega_{\overline{\supp(\omegahat^{\operatorname{JP}(r)}_{\lambda})}}}_1 + \norm{\beta_{\supp(\delta)}}_1 + \norm{\omega_{\supp(\tilde{\delta})}}_1
\\
&=
\norm{\beta_{\Sbar}}_1 + \norm{\omega_{\Tbar}}_1,
\end{aligned}
\end{equation*}
where we have used Lemma~\ref{lemma:identifiability_JP} and the fact that $(\betahat^{\operatorname{JP}(r)}_{\lambda}, \omegahat^{\operatorname{JP}(r)}_{\lambda})$ is identifiable with respect to $X$ and $\lambda$ in the last inequality, and~(\ref{supp_decomp}) for the last equality. Lemma~\ref{lemma:identifiability_JP} concludes our proof.
\end{proof}

\section{Proof of Proposition~\ref{thm:TJP}}
\label{app:proof_thm_TJP}

\begin{proof}[Proof of Proposition~\ref{thm:TJP}]
We define $\tilde{X} := \begin{bmatrix}X & \sqrt{n} I_n\end{bmatrix},$ and
$
\tilde{\nu} = \begin{bmatrix} \betatilde \\ \omegatilde \end{bmatrix} := \tilde{X}^T (\tilde{X} \tilde{X}^T)^{-1} \epsilon.
$ We will assume for now that the following properties hold.
\begin{enumerate}[a)]
\item \label{extstableNSP} Every pair $(\beta, \omega)$ such that $X \beta + \sqrt{n} \omega = 0$ satisfies
\[
\norm{\beta_{S^0}}_1 + \lambda \norm{\omega_{T^0}}_1 \leq \frac{1}{3} (\norm{\beta_\Sbar}_1 + \lambda \norm{\omega_\Tbar}_1),
\]
\item \label{l2bound_nutilde} $\norm{\tilde{\nu}}_2 \leq \frac{\sqrt{2}\sigma}{\left( \frac{\lambdamin(\Sigma)}{4} (\sqrt{p/n} - 1)^2 + 1 \right)^{1/2}}.$
\end{enumerate}
Since $\tilde{X} \tilde{\nu} = X \betatilde + \sqrt{n} \omegatilde = \epsilon,$ one can rewrite model~(\ref{model:sparsecorruption}) as
\begin{equation*}
y = X(\beta^0 + \betatilde) + \sqrt{n} (\omega^0 + \omegatilde).
\end{equation*}
By property~\ref{extstableNSP}) and Lemma~\ref{lemma:stability_JP} below, one has
\begin{equation}
\norm{\betahat^{\operatorname{JP}}_\lambda - (\beta^0 + \betatilde)}_1 + 
\lambda \norm{\omegahat^{\operatorname{JP}}_\lambda - (\omega^0 + \omegatilde)}_1
\leq 
4 (\norm{\betatilde_{\Sbar}}_1 + \lambda \norm{\omegatilde_{\Tbar}}_1),
\end{equation}
and therefore $\norm{\betahat^{\operatorname{JP}}_\lambda - (\beta^0 + \betatilde)}_1
\leq
4 (\norm{\betatilde}_1 + \lambda \norm{\omegatilde}_1).$
Consequently, for any $j \in [p]$ one has
\begin{equation*}
\begin{aligned}
\abs{\betahatJP_{\lambda, j} - \beta^0_j}
&\leq
\abs{\betahatJP_{\lambda, j} - (\beta^0_j + \betatilde_j)} + \abs{\betatilde_j}
\leq
\norm{\betahat^{\operatorname{JP}}_\lambda - (\beta^0 + \betatilde)}_1 + \norm{\betatilde}_1
\\
&\leq 
4 (\norm{\betatilde}_1 + \lambda \norm{\omegatilde}_1)  + \norm{\betatilde}_1
\leq
5 (\norm{\betatilde}_1 +  \lambda \norm{\omegatilde}_1)
\\
&\leq
5 \max\{1, \lambda\} (\norm{\betatilde}_1 + \norm{\omegatilde}_1)
=
5 \max\{1, \lambda\} \norm{\tilde{\nu}}_1
\\
&\leq
5 \max\{1, \lambda\} \sqrt{p+n} \norm{\tilde{\nu}}_2
\leq
\frac{5 \sqrt{2} \max\{1, \lambda\} \sigma \sqrt{p+n}}{( \frac{\lambdamin(\Sigma)}{4} (\sqrt{p/n} - 1)^2 + 1 )^{1/2}}
\end{aligned}
\end{equation*}
where we have used property~\ref{l2bound_nutilde}) in the last inequality. Now setting
\[
\tau := \frac{5 \sqrt{2} \max\{1, \lambda\} \sigma \sqrt{p+n}}{( \frac{\lambdamin(\Sigma)}{4} (\sqrt{p/n} - 1)^2 + 1 )^{1/2}},
\]
one gets 
\begin{equation} \label{bound:tau}
\abs{\betahatJP_{\lambda, j} - \beta^0_j} \leq \tau
\end{equation}
for every $j \in [p].$ If $j \in \Sbar,$ we have $\abs{\betahatJP_{\lambda, j}} \leq \tau,$ hence $\betahatTJP_{(\lambda, \tau), j} = 0.$ If $j \in S^0,$ assumption~(\ref{beta_min}) implies $\abs{\beta^0_j} > 2 \tau,$ which together with~(\ref{bound:tau}) gives $\sign(\betahatTJP_{(\lambda, \tau), j}) = \sign(\beta^0_j).$

It remains to prove that properties~\ref{extstableNSP}) and \ref{l2bound_nutilde}) hold with high probability. First, Lemma~1 in \citet{nguyen2013a}, implies that with probability greater than $1 - c e^{-c' n}$ the matrix $X$ satisfies the extended restricted eigenvalue property
\begin{equation} \label{extRE}
\begin{aligned}
\norm{\beta_\Sbar}_1 + \lambda \norm{\omega_\Tbar}_1 &\leq 3 (\norm{\beta_{S^0}}_1 + \lambda \norm{\omega_{T^0}}_1)
\\ &\Downarrow
\\
\frac{1}{n} \norm{X \beta + \sqrt{n} \omega}_2^2 &\geq \gamma^2 (\norm{\beta}_2^2 + \norm{\omega}_2^2),
\end{aligned}
\end{equation}
with $\gamma^2 = \frac{\min\{\lambdamin(\Sigma), 1\}}{16^2}.$ Property~(\ref{extRE}) clearly implies~\ref{extstableNSP}). Finally, Lemma~\ref{lemma:l2norm_nutilde} below proves that~\ref{l2bound_nutilde}) holds with probability at least $1 - 1.14^{-n} - 2e^{-\frac{1}{8}(\sqrt{p} - \sqrt{n})^2},$ which concludes our proof. 
\end{proof}

\begin{lemma} \label{lemma:stability_JP}
Assume that for some sets $S^0 \subset [p]$ and $T^0 \subset [n],$ and some constant $\rho \in (0, 1),$ the matrix $X \in \R^{n \times p}$ satisfies
\begin{equation} \label{ext_stableNSP}
\norm{\beta_{S^0}}_1 + \lambda \norm{\omega_{T^0}}_1 \leq \rho (\norm{\beta_{\Sbar}}_1 + \lambda \norm{\omega_{\Tbar}}_1),
\end{equation}
for every pair $(\beta, \omega) \in \R^p \times \R^n$ such that $X \beta + \sqrt{n} \omega = 0.$ Then for every pair $(\betatilde, \omegatilde) \in \R^{p} \times \R^n,$ the solution $(\betahatJP_{\lambda}, \omegahatJP_{\lambda})$ to JP~(\ref{opt:JP_newscaling}) with $y = X \betatilde + \sqrt{n} \omegatilde$ satisfies
\begin{equation*} 
\norm{\betahatJP_{\lambda} - \betatilde}_1 + \lambda \norm{\omegahatJP_{\lambda} - \omegatilde}_1
\leq
\frac{2(1+\rho)}{1-\rho} (\norm{\betatilde_{\Sbar}}_1 + \lambda \norm{\omegatilde_{\Tbar}}_1).
\end{equation*}
\end{lemma}
\begin{proof}
This proof is a simple extension of the one of Theorem 4.14 in \citet{foucart2013}. Let us consider $y = X \betatilde + \sqrt{n} \omegatilde$ for an arbitrary pair $(\betatilde, \omegatilde),$ and define
$
\beta' := \betahatJP_{\lambda} -  \betatilde
$ and
$\omega' := \omegahatJP_{\lambda} - \omegatilde.$
Clearly $X \beta' + \sqrt{n} \omega' = 0,$ so by~(\ref{ext_stableNSP}),
\begin{equation} \label{ineq:extstableNSP}
\norm{\beta'_{S^0}}_1 + \lambda \norm{\omega'_{T^0}}_1
\leq
\rho (\norm{\beta'_{\Sbar}}_1 + \lambda \norm{\omega'_{\Tbar}}_1).
\end{equation}
We also have
\begin{equation*}
\begin{aligned}
\norm{\betatilde}_1 + \lambda \norm{\omegatilde}_1
&=
\norm{\betatilde_{S^0}}_1 + \norm{\betatilde_{\Sbar}}_1
+
\lambda (\norm{\omegatilde_{T^0}}_1 + \norm{\omegatilde_{\Tbar}}_1)
\\
&=
\norm{\betahatJP_{\lambda, S^0} - \beta'_{S^0}}_1 + \norm{\betatilde_{\Sbar}}_1
+
\lambda (\norm{\omegahatJP_{\lambda, T^0} - \omega'_{T^0}}_1 + \norm{\omegatilde_{\Tbar}}_1)
\\
& \leq
 \norm{\betahatJP_{\lambda, S^0}}_1 + \norm{\beta'_{S^0}}_1  + \norm{\betatilde_{\Sbar}}_1
+
\lambda (\norm{\omegahatJP_{\lambda, T^0}}_1 +  \norm{\omega'_{T^0}}_1 + \norm{\omegatilde_{\Tbar}}_1),
\end{aligned}
\end{equation*}
and
\begin{equation*}
\norm{\beta'_{\Sbar}}_1 + \lambda \norm{\omega'_{\Tbar}}_1
\leq
(\norm{\betahatJP_{\lambda, \Sbar}}_1 + \norm{\betatilde_{\Sbar}}_1) + \lambda (\norm{\omegahatJP_{\lambda, \Tbar}}_1 + \norm{\omegatilde_{\Tbar}}_1).
\end{equation*}
Adding the last two inequalities yields
\begin{equation*}
\begin{aligned}
\norm{\beta'_{\Sbar}}_1 + \lambda \norm{\omega'_{\Tbar}}_1 + \norm{\betatilde}_1 + \lambda \norm{\omegatilde}_1
&\leq
\norm{\betahatJP_{\lambda}}_1 + \norm{\beta'_{S^0}}_1  +  2 \norm{\betatilde_{\Sbar}}_1
\\
& \quad +
\lambda (\norm{\omegahatJP_{\lambda}}_1 + \norm{\omega'_{T^0}}_1 +  2\norm{\omegatilde_{\Tbar}}_1),
\end{aligned}
\end{equation*}
and rearranging terms gives
\begin{equation*}
\begin{aligned}
\norm{\beta'_{\Sbar}}_1 + \lambda \norm{\omega'_{\Tbar}}_1
& \leq
(\norm{\betahatJP_{\lambda}}_1 + \lambda \norm{\omegahatJP_{\lambda}}_1) - (\norm{\betatilde}_1 + \lambda \norm{\omegatilde}_1)
\\
& \quad + (\norm{\beta'_{S^0}}_1 + \lambda \norm{\omega'_{T^0}}_1) + 2(\norm{\betatilde_{\Sbar}}_1 + \lambda \norm{\omegatilde_{\Tbar}}_1).
\end{aligned}
\end{equation*}
Using~(\ref{ineq:extstableNSP}) and the fact that $\norm{\betahatJP_{\lambda}}_1 + \lambda \norm{\omegahatJP_{\lambda}}_1 \leq  \norm{\betatilde}_1 + \lambda \norm{\omegatilde}_1$ by minimality of the JP solution, we get
\begin{equation*}
\norm{\beta'_{\Sbar}}_1 + \lambda \norm{\omega'_{\Tbar}}_1
\leq
\rho(\norm{\beta'_{\Sbar}}_1 + \lambda \norm{\omega'_{\Tbar}}_1) 
+ 2(\norm{\betatilde_{\Sbar}}_1 + \lambda \norm{\omegatilde_{\Tbar}}_1),
\end{equation*}
hence
\begin{equation} \label{ineq:complement}
\norm{\beta'_{\Sbar}}_1 + \lambda \norm{\omega'_{\Tbar}}_1
\leq
\frac{2}{1-\rho} (\norm{\betatilde_{\Sbar}}_1 + \lambda \norm{\omegatilde_{\Tbar}}_1).
\end{equation}
Now inequality~(\ref{ineq:extstableNSP}) also implies
\begin{equation} \label{ineq:extstableNSP2}
\begin{aligned}
\norm{\beta'}_1 + \lambda \norm{\omega'}_1
&=
\norm{\beta'_{S^0}}_1 + \lambda \norm{\omega'_{T^0}}_1 + \norm{\beta'_{\Sbar}}_1 + \lambda \norm{\omega'_{\Tbar}}_1
\\
&\leq
(1 + \rho) (\norm{\beta'_{\Sbar}}_1 + \lambda \norm{\omega'_{\Tbar}}_1)
\end{aligned}
\end{equation}
and continuing~~(\ref{ineq:extstableNSP2}) with~(\ref{ineq:complement}) gives the desired inequality.
\end{proof}
%%%%%%%%%%%%%%%%%%%%%%%%%%%%%%%%%%%%%%%%%%%%%%

\begin{lemma} \label{lemma:l2norm_nutilde}
Let $\Xtilde := \begin{bmatrix} X & \sqrt{n} I_n \end{bmatrix}.$ Under assumptions~\ref{assumption:Gaussdesign}), \ref{assumption:boundedsigmamin}), \ref{assumption:var1}) and \ref{assumption:Gaussiannoise}), 
\[
\norm{\Xtilde^T (\Xtilde \Xtilde^T)^{-1} \epsilon}_2 \leq \frac{\sqrt{2}\sigma}{( \frac{\lambdamin(\Sigma)}{4} (\sqrt{p/n} - 1)^2 + 1 )^{1/2}},
\]
with probability at least $1 - 1.14^{-n} - 2e^{-\frac{1}{8}(\sqrt{p} - \sqrt{n})^2}.$
\end{lemma}
\begin{proof}
We have 
\[
\norm{\Xtilde^T (\Xtilde \Xtilde^T)^{-1} \epsilon}_2^2 
= 
\epsilon^T (\Xtilde \Xtilde^T)^{-1} \epsilon 
\leq 
\frac{\norm{\epsilon}_2^2}{\lambdamin(\Xtilde \Xtilde^T)}
= 
\frac{\sigma \norm{\tfrac{1}{\sigma}\epsilon}_2^2}{\lambdamin(\Xtilde \Xtilde^T)}.
\]
Since $\norm{\tfrac{1}{\sigma}\epsilon}_2^2 \sim \chi^2_n,$ it is upper bounded by $2n$ with probability larger than $1 - 1.14^{-n}$ (a corollary of Lemma 1 in \citet{laurent2000}). So 
\begin{equation} \label{bound:nutilde}
\p\left(\norm{\tilde{\nu}}_2 \leq \frac{\sqrt{2n} \sigma}{\sigmamin(\Xtilde)}\right) \geq 1 - 1.14^{-n}.
\end{equation}
Let us now bound $\sigmamin(\Xtilde).$ One has
\begin{equation} \label{sigmamin_relation}
\begin{aligned}
\sigmamin^2(\Xtilde) 
=
\lambdamin(\Xtilde \Xtilde^T)
=
\lambdamin(XX^T + n I_n)
=
\sigmamin^2(X) + n .
\end{aligned}
\end{equation}
One can write $X = G \Sigma^{1/2}$ where $G \in \R^{n \times p}$ with $G_{ij} \overset{\textrm{i.i.d.}}{\sim} N(0, 1),$ thus
\begin{equation} \label{ineq:sigmamin}
\sigmamin(X) \geq \sigmamin(G) \sigmamin(\Sigma^{1/2}) = \sigmamin(G) \sqrt{\lambdamin(\Sigma)}.
\end{equation}
Now it is known (see \citet{rudelson2010}, eq. (2.3)) that
\[
\sigmamin(G) \geq \frac{1}{2}(\sqrt{p} - \sqrt{n}) = \frac{\sqrt{n}}{2} (\sqrt{p/n} - 1)
\]
with probability at least $1 - 2e^{-\frac{1}{8} (\sqrt{p} - \sqrt{n})^2}.$ Together with~(\ref{sigmamin_relation}) and (\ref{ineq:sigmamin}) this gives
\begin{equation*}
\p\left(\sigmamin(\Xtilde) \geq \left(\frac{n \lambdamin(\Sigma)}{4} (\sqrt{p/n} - 1)^2 + n\right)^{1/2} \right) \geq 1 - 2e^{-\frac{1}{8} (\sqrt{p} - \sqrt{n})^2}.
\end{equation*}
With~(\ref{bound:nutilde}), this implies
\begin{equation*}
\p\left(\norm{\tilde{\nu}}_2 \leq  \frac{\sqrt{2}\sigma}{( \frac{\lambdamin(\Sigma)}{4} (\sqrt{p/n} - 1)^2 + 1 )^{1/2}}\right) \geq 1 - 1.14^{-n} - 2e^{-\frac{1}{8}(\sqrt{p} - \sqrt{n})^2}.
\end{equation*}
\end{proof}

\section{Additional numerical experiments}

\subsection{Comparison of imputation methods}\label{sec:othernumexp_imp}

{We perform other numerical experiments for Robust Lasso-Zero changing the initial naive imputation in order to show if it affects its performances. We compare two different ways of imputing beforehand of the Robust Lasso methodology when covariates are centered at $(1,\hdots, 1)$: either an imputation by the empirical mean of the observed entries is considered (and thus should be around 1), or an arbitrary imputation by zero. 
The results are shown in Figure \ref{fig:imputations} for non-correlated features in (a), and correlated ones in (b): no clear conclusion can be actually drawn.

Indeed, in Figure \ref{fig:imputations} (a) (with non-correlated variables), for $s=3$ (left), both imputations lead to similar results; but for $s=10$ (right), the mean imputation outperforms  the zero-imputation with 5\% MNAR values (top right), and the contrary holds for 10\% of MCAR values (bottom right).

In Figure \ref{fig:imputations} (b) (with correlated design),  the imputation strategies may lead to the same results (see for instance when $s=3$ with MCAR values (left), or for $s=10$ with 5\% MCAR values (top right), or even for $s=10$ and 20\% MNAR values (bottom right)), which is the most challenging setting in these simulations. 
%corresponds to the most difficult case when both imputation strategies leads to a null probability of sign recovery. 
In other cases (i.e.\ all the settings with MNAR data and 20\% MCAR data (bottom right)), the mean imputation outperforms the zero imputation.}

\begin{figure}[h!]
%\vspace{-2cm}
%\hspace{-2cm}
\footnotesize 
\begin{tabular}{cc}
\includegraphics[width=0.47\textwidth]{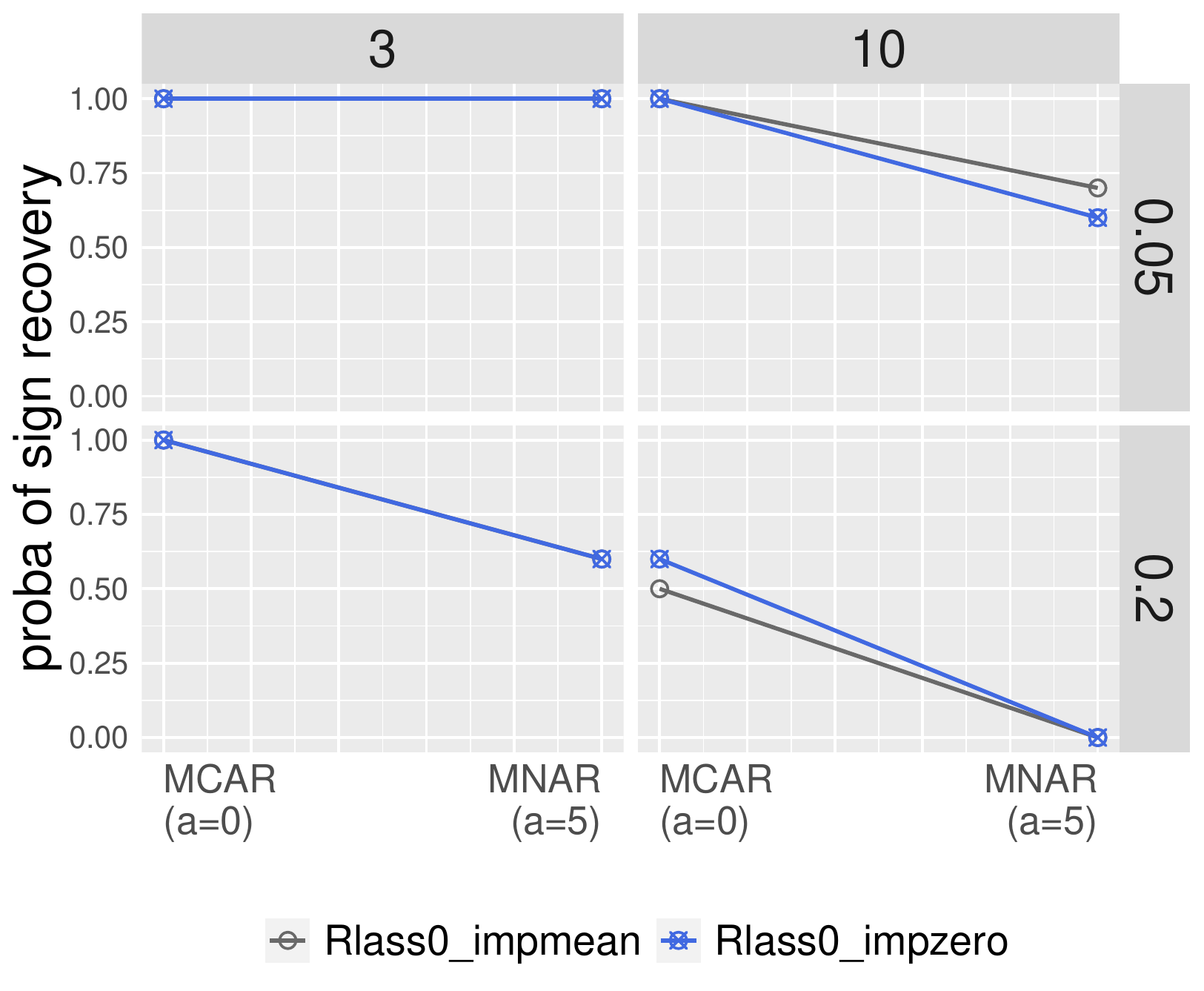} &
\includegraphics[width=0.47\textwidth]{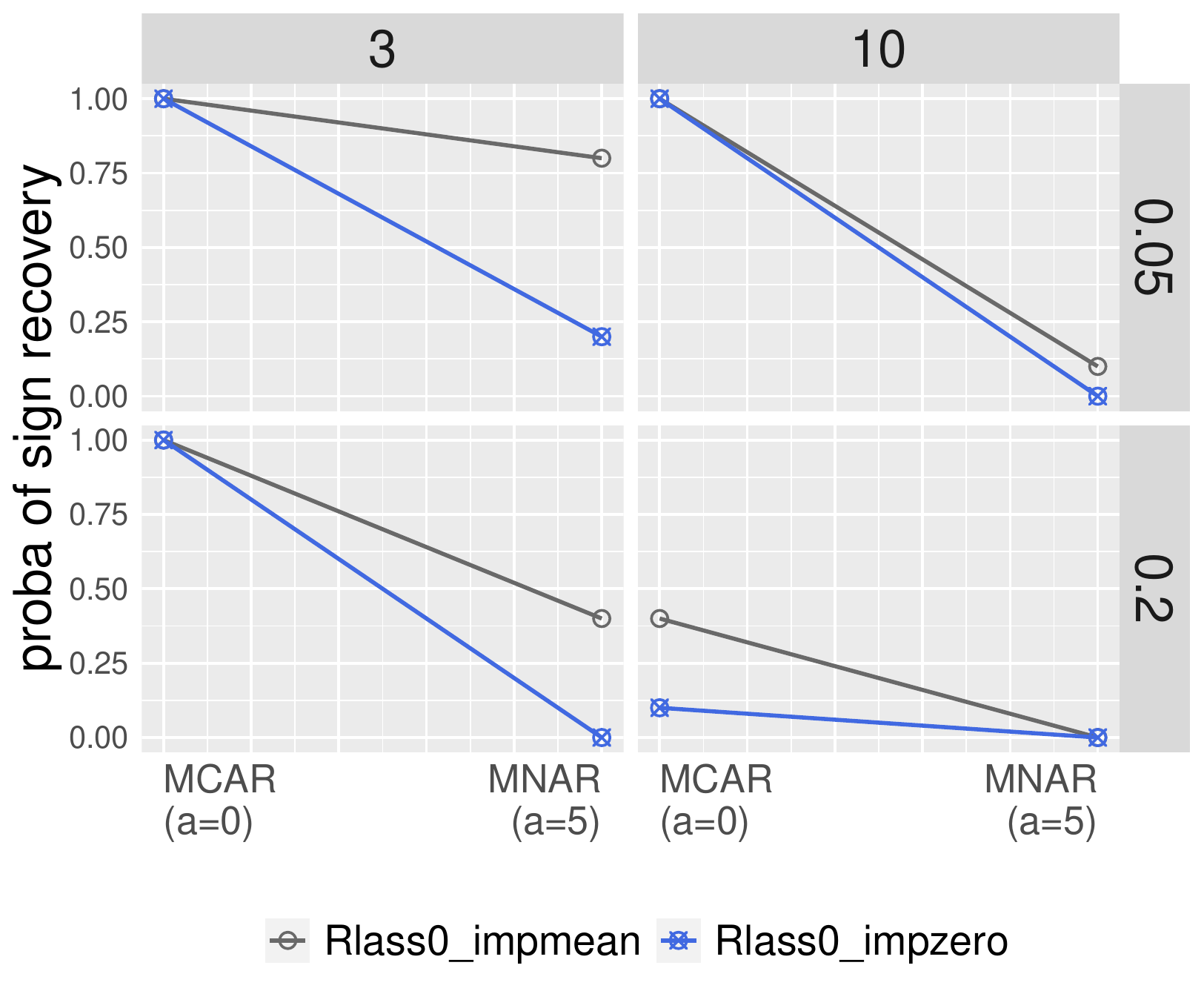} \\
(a) Support recovery, $\rho=0$  & (b) Sign recovery, $\rho=0.5$
\end{tabular}
\caption{\label{fig:imputations} \footnotesize Support and sign recovery for Robust Lass0 with mean imputation (Rlass0\_impmean) and imputation by zeros (Rlass0\_impzero), when $n=100$, $p=200$, $\sigma=0.5$ and the covariates are centered at $(1,\hdots, 1)$.  An $s$-oracle tuning for the threshold is used  for different sparsity levels $s=3$ and $s=10$ (subplots columns), and different proportions of missing values $5\%$ or $20\%$ (subplots rows), and different missing data mechanisms (MCAR vs MNAR).}
\end{figure}

\subsection{Comparison of the Robust Lasso-Zero with other estimators}\label{sec:othernumexp_othestim}

{For a non-oracle hyperparameter tuning, we compare the Robust Lasso-Zero with the thresholded Robust Lasso proposed in \cite{nguyen2013a} and the thresholded lasso in \cite{pokarowski2019improving}. 

    For the Robust Lasso-Zero, we select the threshold $\tau$ by using the quantile universal threshold (QUT) at level $\alpha=0.05$. 
    
    For the thresholded version of the Robust-Lasso and of the Lasso, we used the Generalized Information Criterion (GIC), proposed in the R package \texttt{DMRnet} \cite{pokarowski2019improving}. 
    One could note that this criterion requires the estimation of the noise level, which is not needed for the Robust Lasso-Zero using quantile universal threshold instead. Besides, these methods require the choice of the parameter $c$ controlling the amount of penalization in regard of the model complexity in the GIC, which means that these methods require one hyperparameter for the thresholding tuning, which is one more than needed for the Robust Lasso-Zero with QUT. 
    It is recommended to use the default value $c=2.5$ for linear regression. Figure \ref{fig:threshold_auto} (a) shows that if we choose such a parameter, both thresholded versions of the Robust Lasso and the Lasso do not recover the sign. In Figure \ref{fig:threshold_auto} (b) and (c), Robust Lasso-Zero with QUT outperforms the other methods in terms of s-TPR (the signed True Positive Rate) and s-FDR (the signed False Discovery Rate). The difficulty in choosing $c$ may be due to the \texttt{DMRnet} code and not to the method, but this hyperparameter choice remains under the responsibility of the user, which shows the drawback of the information criteria involving more hyperparameters. 
    In Figure \ref{fig:threshold}, we have also compared the Robust Lasso-Zero with the thresholded versions of the Robust Lasso and the Lasso, but by manually tuning an optimal parameter controlling the penalization in the GIC. 
    
    Note here that the thresholded version of the Robust Lasso always outperforms the thresholded lasso, which was expected, because it allows to account for the corruptions. 
    
    In terms of sign recovery, when the covariates are correlated (Figure \ref{fig:threshold} (a)) or not (Figure \ref{fig:threshold} (b)), the Robust Lasso-Zero is at least equivalent to the two other strategies or outperforms them in the MNAR setting when $s=10$ (right) and in the MCAR setting when $s=10$ (right) (except in the correlated case for $5\%$ of MCAR values). The Robust Lasso-Zero should be then the recommended strategy handling difficult cases in terms of sign recovery, i.e.\ when the percentage of missing values increases, when the mechanism is MNAR and when the covariates are correlated. 
    
    Besides, even if the Robust Lasso Zero may provide lower probabilities of sign recovery than the thresholded versions of the (Robust) Lasso in some particular cases (see for instance in Figure \ref{fig:threshold} (b) (left) for $s=3$, $5\%$ or with $20\%$ MCAR values), 
    the Robust Lasso Zero always remains competitive in terms of s-FDR (see Figure \ref{fig:threshold} (e) and  (f)) and always give the best s-TPR (see Figure \ref{fig:threshold} (c) and (d)).
    }

        \begin{figure}[H]
    %\vspace{-2cm}
    %\hspace{-2cm}
        \footnotesize
    \begin{tabular}{cc}
    \includegraphics[width=0.47\textwidth]{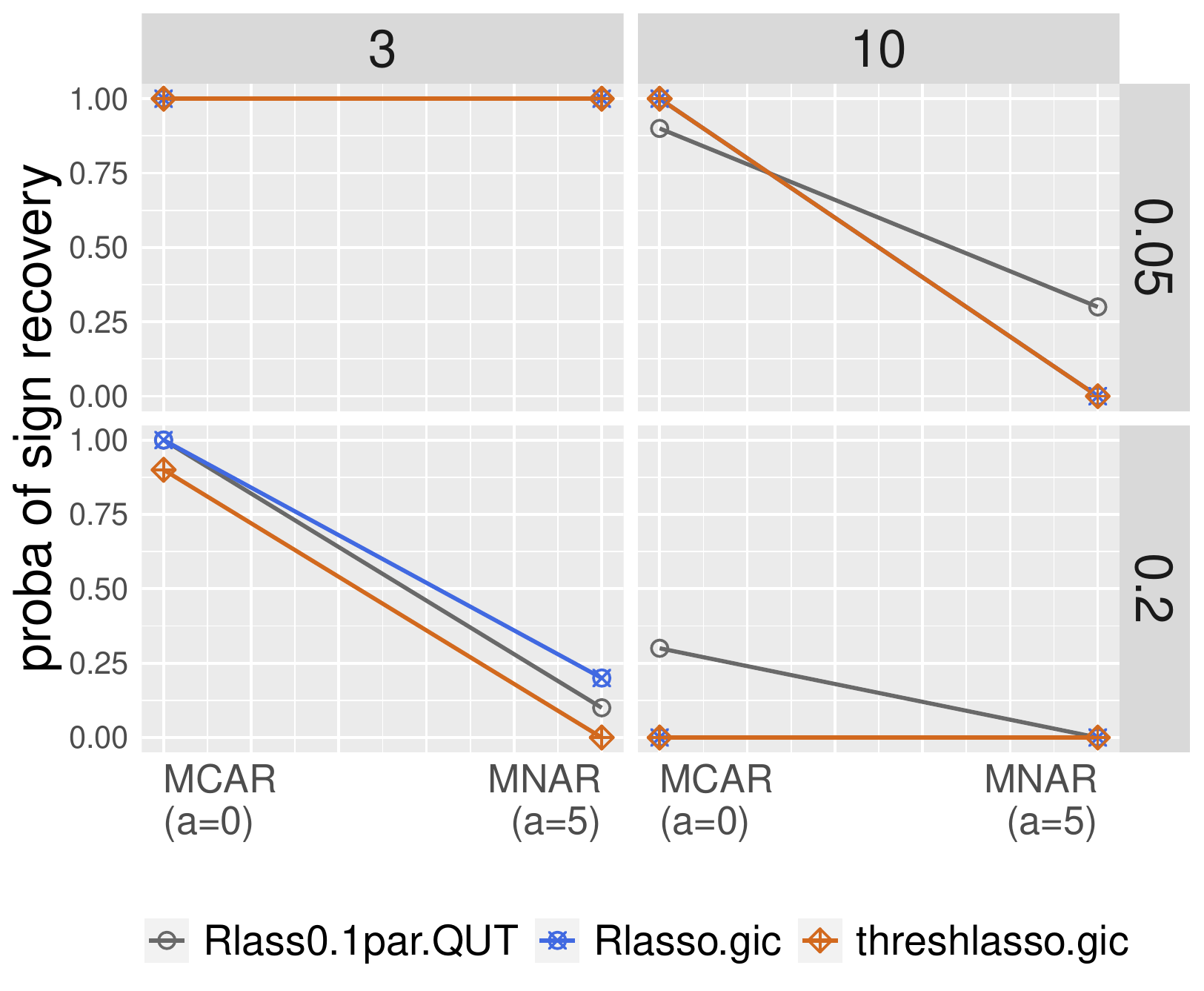} &
    \includegraphics[width=0.47\textwidth]{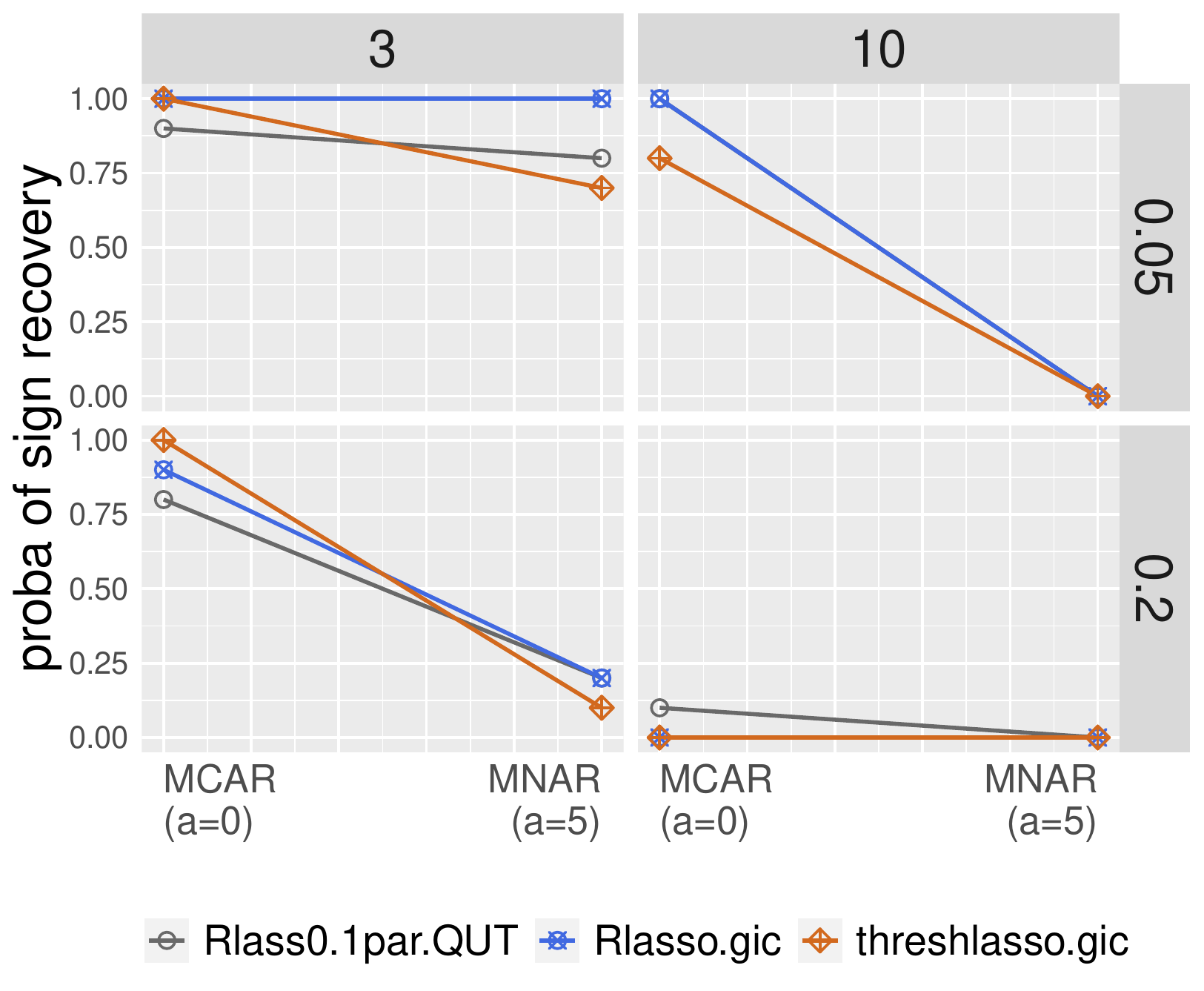} \\
    (a) Sign recovery, $\rho=0$  & (b) Sign recovery, $\rho=0.5$ \\
    \includegraphics[width=0.47\textwidth]{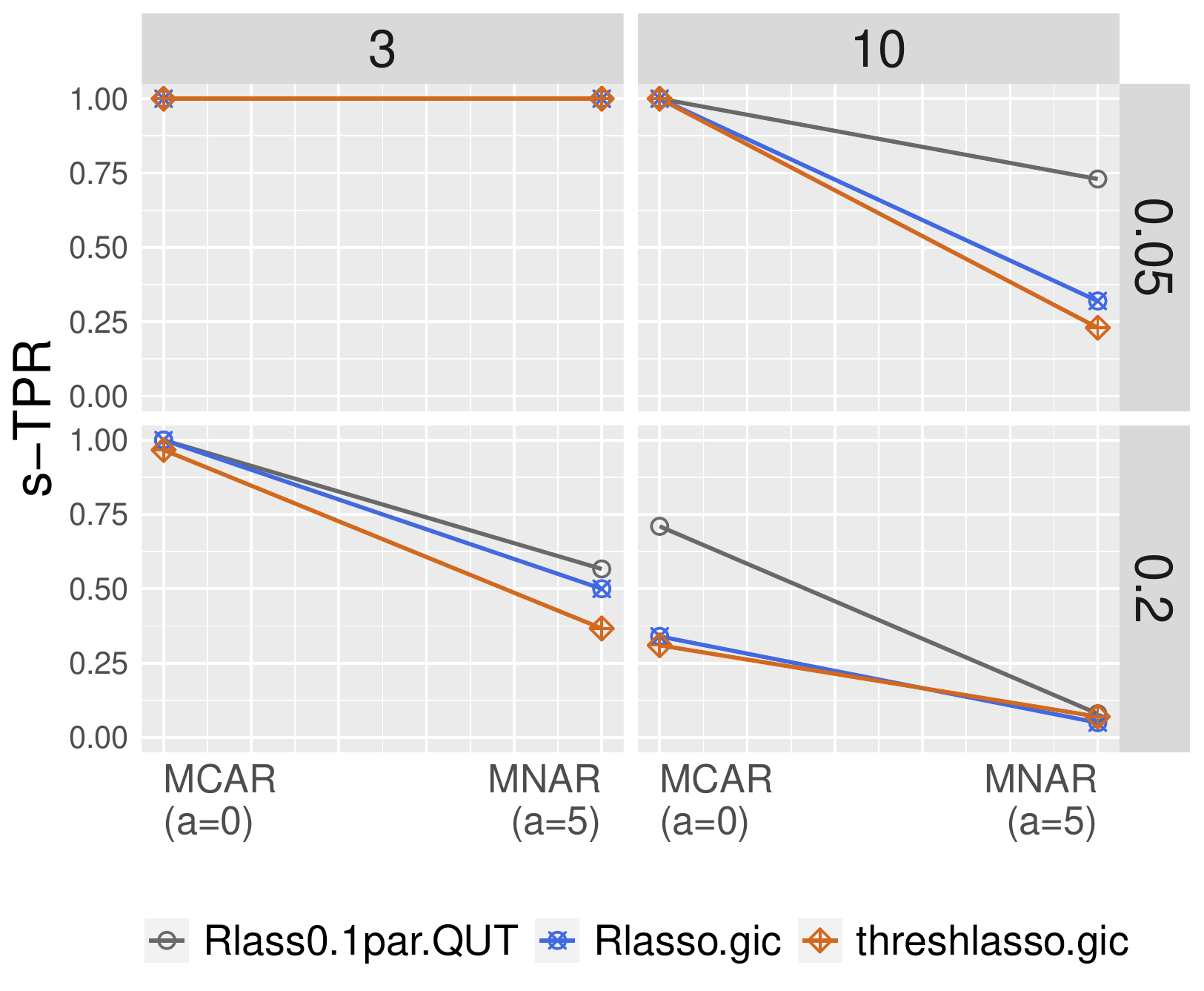} &
    \includegraphics[width=0.47\textwidth]{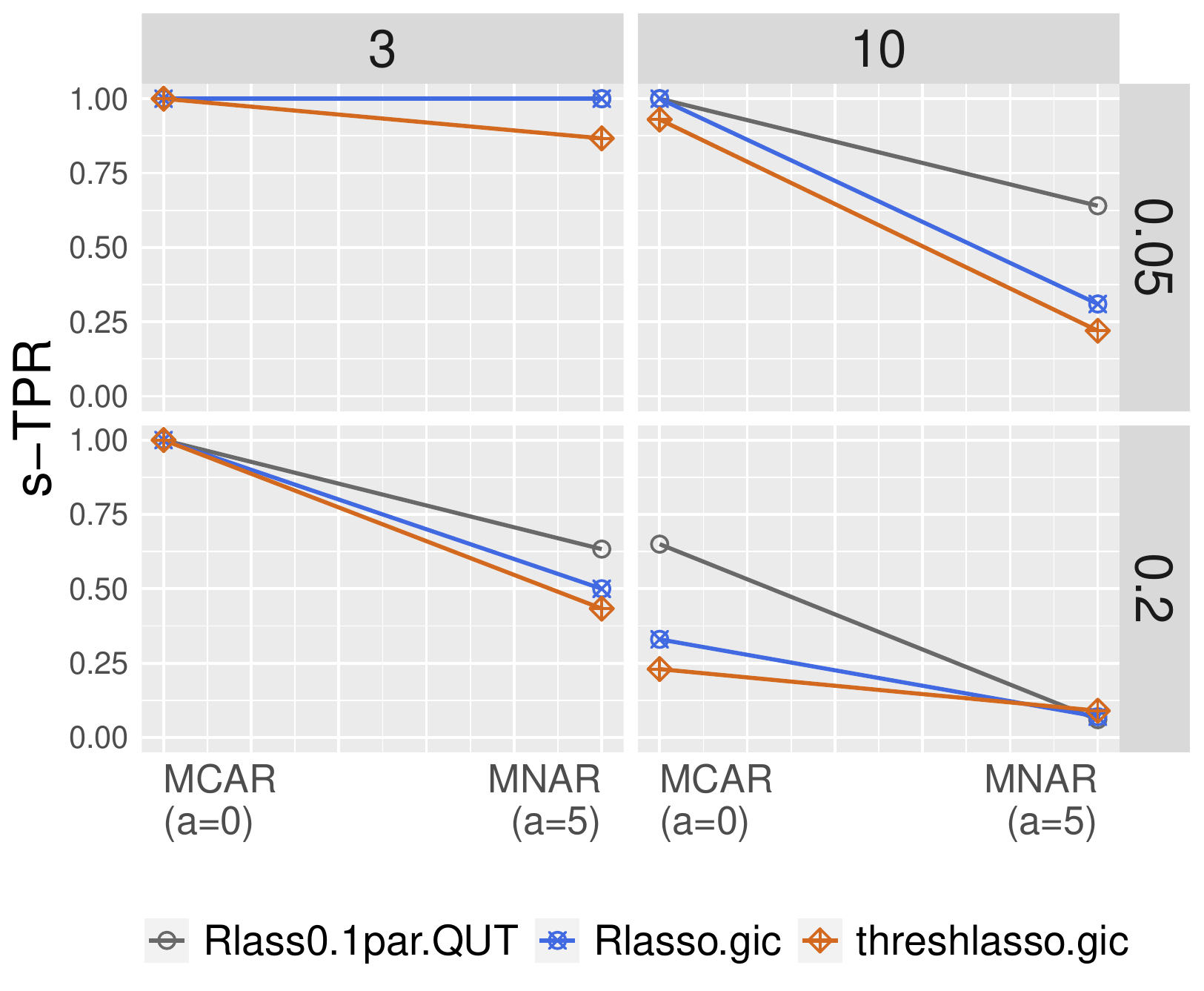} \\
    (c) s-TPR, $\rho=0$  & (d) s-TPR, $\rho=0.5$ \\\includegraphics[width=0.47\textwidth]{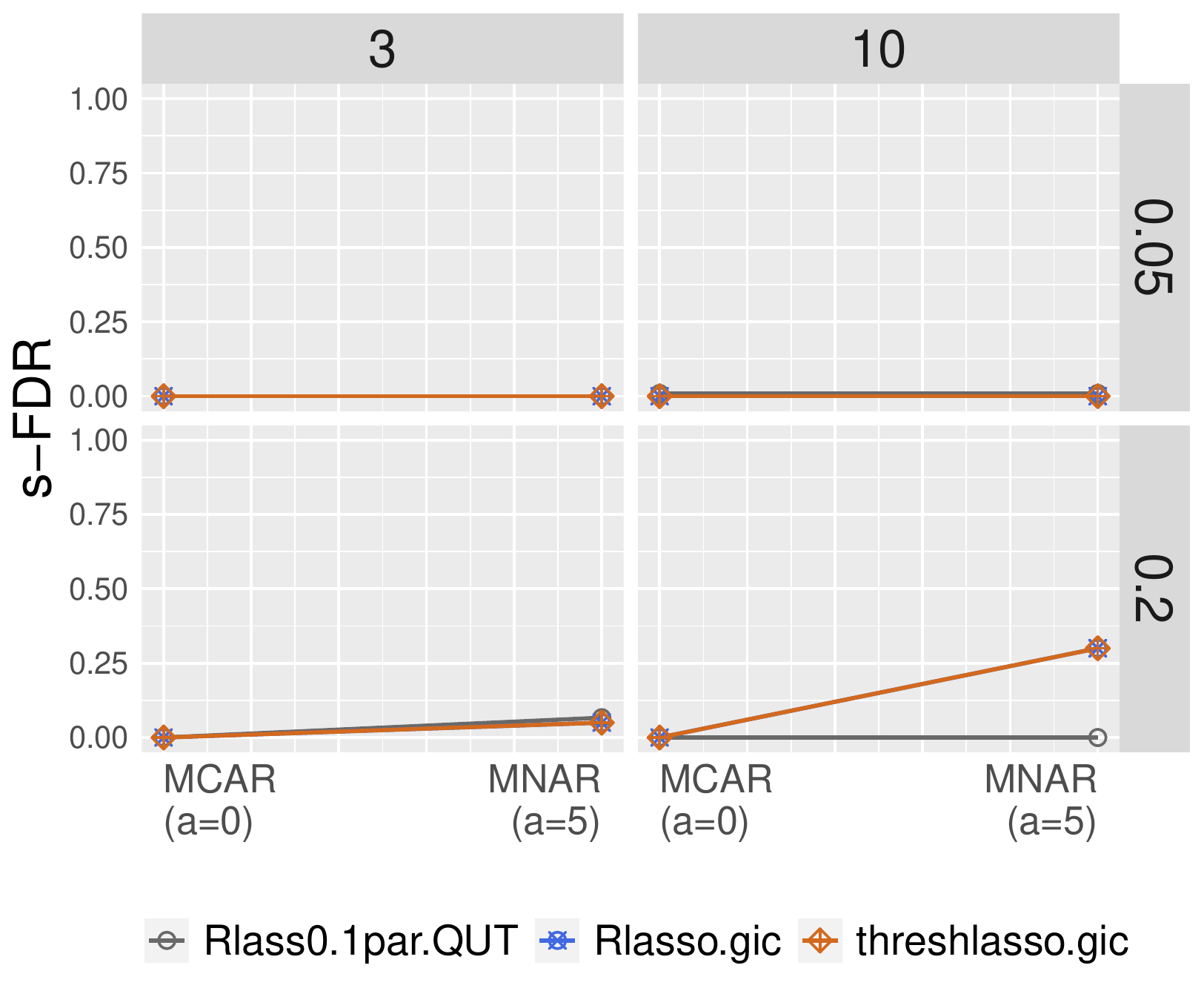} &
    \includegraphics[width=0.47\textwidth]{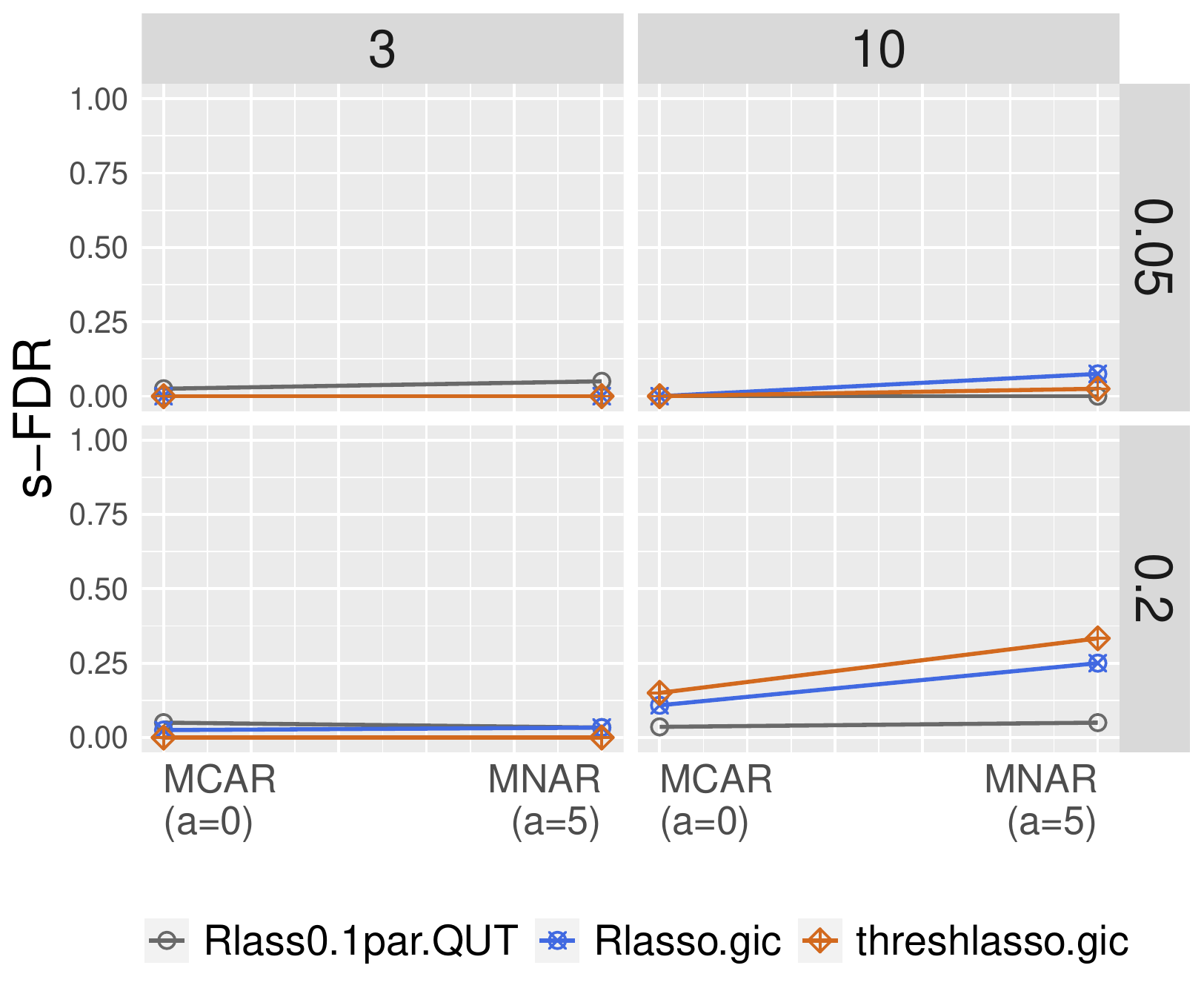} \\
    (e) s-FDR, $\rho=0$ & (f) s-FDR, $\rho=0.5$
    \end{tabular}
    \caption{\label{fig:threshold} \footnotesize Sign recovery, s-TPR and s-FDR for the Robust Lasso-Zero, a thresholded version of the Robust Lasso and a thresholded version of the Lasso, when $n=100$, $p=200$, $\sigma=0.5$ and the covariates are centered at $(1,\hdots, 1)$.  To tune the threshold, we use QUT for the Robust Lasso-Zero and GIC for the thresholded versions of both the Robust Lasso and the Lasso (the parameter $c$ is chosen scaling as $n$). Different sparsity levels $s=3$ and $s=10$ (subplots columns), different proportions of missing values $5\%$ or $20\%$ (subplots rows), and different missing data mechanisms (MCAR vs MNAR) are considered.}
    \end{figure}
    
    \begin{figure}[H]
    %\vspace{-2cm}
    %\hspace{-2cm}        
    \footnotesize
    \begin{tabular}{cc}
    \includegraphics[width=0.47\textwidth]{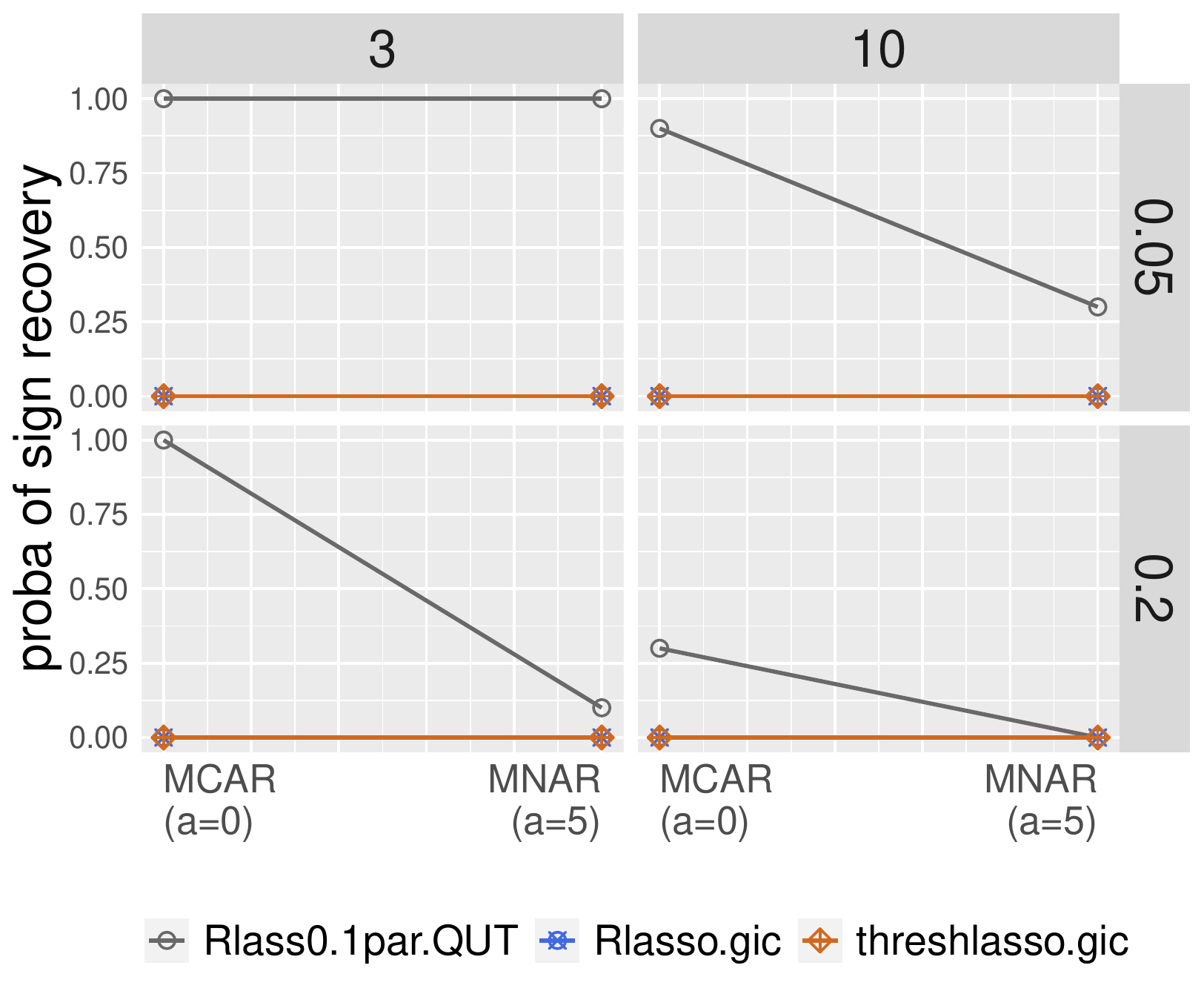} &
    \includegraphics[width=0.47\textwidth]{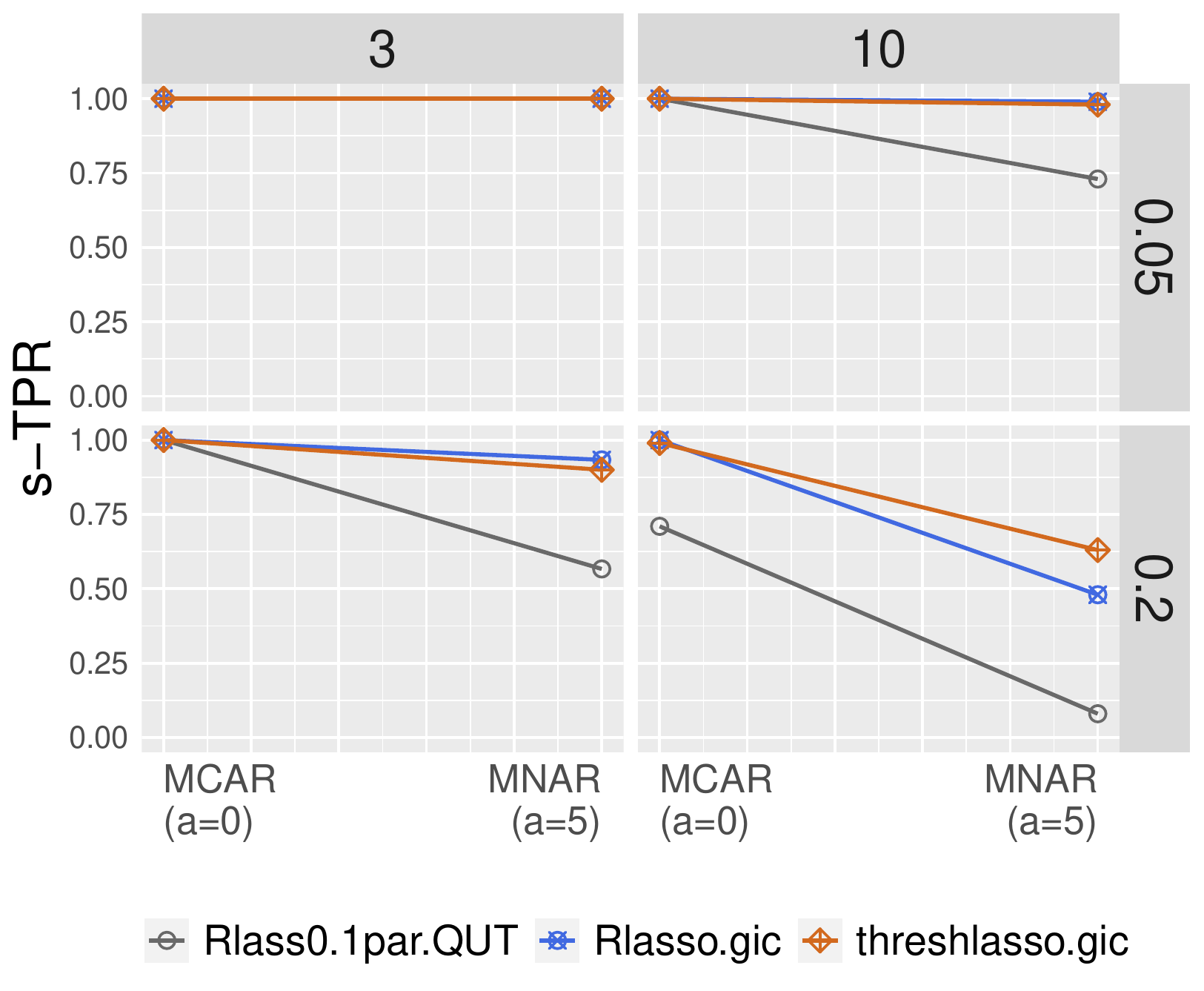} \\
    (a) Sign recovery, $\rho=0$   & (b) s-TPR, $\rho=0$ \\
    \includegraphics[width=0.47\textwidth]{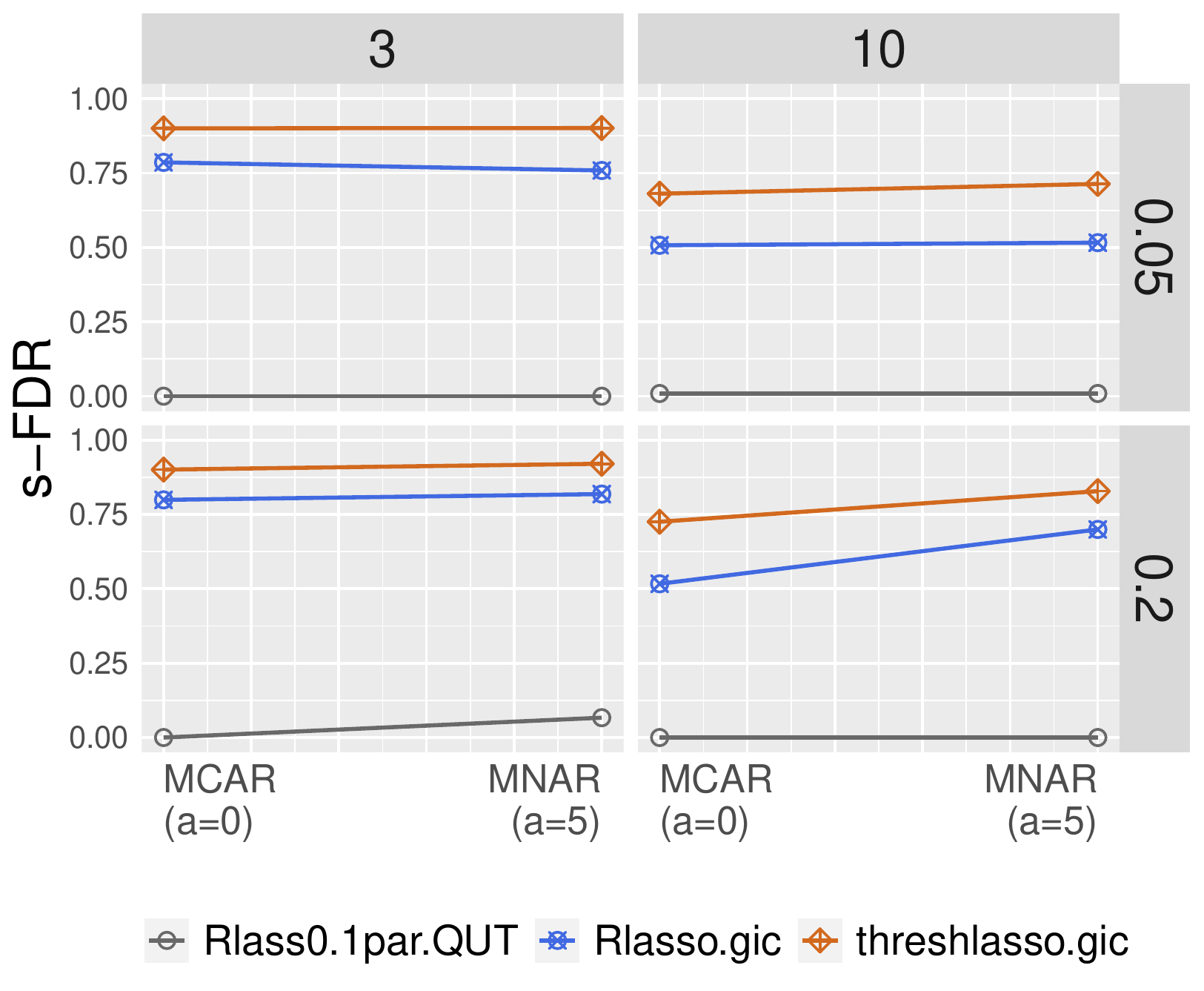} & \\
    (c) s-FDR, $\rho=0$
    \end{tabular}
    \caption{\label{fig:threshold_auto} \footnotesize Sign recovery, s-TPR and s-FDR for the Robust Lasso-Zero, a thresholded version of the Robust Lasso and a thresholded version of the Lasso, when $n=100$, $p=200$, $\sigma=0.5$ and the covariates are centered at $(1,\hdots, 1)$.  To tune the threshold, we use QUT for the Robust Lasso-Zero and GIC for the thresholded versions of both the Robust Lasso and the Lasso (we take for $c$ the default value i.e.\ $c=2.5$). Different sparsity levels $s=3$ and $s=10$ (subplots columns), different proportions of missing values $5\%$ or $20\%$ (subplots rows), and different missing data mechanisms (MCAR vs MNAR) are considered.}
    \end{figure}
    
\section{Variables in the Traumabase dataset}
\label{sec:variablestraumadataset}

The variables of the Traumabase dataset are:
\begin{itemize}
\item \textit{Time.amb}: Time spent in the ambulance, \textit{i.e.}, transportation time from accident site to hospital, in minutes.  
\item \textit{Lactate}: The conjugate base of lactic acid. 
\item \textit{Delta.Hemo}: The difference between the homoglobin on arrival at hospital and that in the ambulance. 
\item \textit{RBC}: A binary index which indicates whether the transfusion of Red Blood Cells Concentrates is performed. 
\item \textit{SI.amb}: Shock index measured on ambulance.
\item \textit{DBP.min}: Minimum value of measured diastolic blood pressure in the ambulance. 
\item \textit{SBP.min}: Minimum value of measured systolic blood pressure in the ambulance. 
\item \textit{HR.max}: Maximum value of measured heart rate in the ambulance. 
\item \textit{VE}: A volume expander is a type of intravenous  therapy that has the function of providing volume for the circulatory system. 
\item \textit{MBP.amb}: Mean arterial pressure measured in the ambulance. 
\item \textit{Temp}: Patient's body temperature. 
\item \textit{SI}: Shock index $SI=HR/SBP$ indicates level of occult shock based on heart rate and systolic blood pressure on arrival at hospital.
\item \textit{MBP}: Mean arterial pressure $MBP=(2DBP+SPB)/3$ is an average blood pressure in an individual during a single cardiac cycle. %, based on systolic blood pressure  and diastolic blood pressure. 
\item \textit{HR}: Heart rate measured on arrival of hospital. 
\item \textit{Age}: Age.
\end{itemize}

\end{document}